\def\llncs{0}
\def\fullpage{1}
\def\anonymous{0}

\def\notxfont{0}
\def\submission{0}
\def\cameraready{0}

\ifnum\submission=1
\def\anonymous{1}
\def\llncs{1}

\else
\fi

\ifnum\cameraready=1
\def\submission{1}
\def\llncs{1}

\def\anonymous{0}
\else
\fi

\ifnum\anonymous=1
\def\llncs{0}

\else
\fi

\ifnum\llncs=1
	\documentclass[envcountsect,a4paper,runningheads,11pt]{llncs}
\else
	\documentclass[letterpaper,hmargin=1.05in,vmargin=1.05in,11pt]{article}
			\ifnum\fullpage=1
		\usepackage{fullpage}
		\fi
\fi

\usepackage[%
  colorlinks=true,
  citecolor=blue,
  pagebackref=true
]{hyperref}

\usepackage{amsmath, amsfonts, amssymb, mathtools,amscd}

\usepackage{amsthm}

\usepackage{fancyhdr}
\usepackage{lmodern}
\usepackage[T1]{fontenc}
\usepackage[utf8]{inputenc}

\usepackage{arydshln} 
\usepackage{url}
\usepackage{ifthen}
\usepackage{bm}
\usepackage{multirow}
\usepackage[dvips]{graphicx}
\usepackage[usenames]{color}
\usepackage{xcolor,colortbl} 
\usepackage{threeparttable}
\usepackage{comment}
\usepackage{paralist,verbatim}
\usepackage{cases}
\usepackage{booktabs}
\usepackage{braket}
\usepackage{cancel} 
\usepackage{ascmac} 
\usepackage{framed}
\usepackage{physics}
\usepackage{authblk}
\usepackage{pifont}
\definecolor{darkblue}{rgb}{0,0,0.6}
\definecolor{darkgreen}{rgb}{0,0.5,0}
\definecolor{maroon}{rgb}{0.5,0.1,0.1}
\definecolor{dpurple}{rgb}{0.2,0,0.65}

\usepackage[capitalise,noabbrev]{cleveref}
\usepackage[absolute]{textpos}
\usepackage[final]{microtype}
\usepackage[absolute]{textpos}
\usepackage{everypage}

\newtheoremstyle{thicktheorem}%
{\topsep}
{\topsep}
{\itshape}{}%
{\bfseries}%
{.}
{ }%
{\thmname{#1}\thmnumber{ #2}%
		\thmnote{ (#3)}%
}

\newtheoremstyle{remark}
{\topsep}
{\topsep}
	{}
	{}
	{}
	{.}
	{ }
	{\textit{\thmname{#1}}\thmnumber{ #2}
			\thmnote{ (#3)}%
	}

\ifnum\llncs=0
	\theoremstyle{thicktheorem}
	\newtheorem{theorem}{Theorem}[section]
	\newtheorem{lemma}[theorem]{Lemma}
        \newtheorem{assumption}[theorem]{Assumption}

	\newtheorem{definition}[theorem]{Definition}
	\newtheorem{game}[theorem]{Game}

	\theoremstyle{remark}
	\newtheorem{claim}[theorem]{Claim}

\else
\fi

	\crefname{theorem}{Theorem}{Theorems}
	\crefname{assumption}{Assumption}{Assumptions}
	\crefname{construction}{Construction}{Constructions}
	\crefname{corollary}{Corollary}{Corollaries}
	\crefname{conjecture}{Conjecture}{Conjectures}
	\crefname{definition}{Definition}{Definitions}
	\crefname{exmaple}{Example}{Examples}
	\crefname{experiment}{Experiment}{Experiments}
	\crefname{counterexample}{Counterexample}{Counterexamples}
	\crefname{lemma}{Lemma}{Lemmata}
	\crefname{observation}{Observation}{Observations}
	\crefname{proposition}{Proposition}{Propositions}
	\crefname{remark}{Remark}{Remarks}
	\crefname{claim}{Claim}{Claims}
	\crefname{fact}{Fact}{Facts}
	\crefname{note}{Note}{Notes}

\ifnum\llncs=1
 \crefname{appendix}{App.}{Appendices}
 \crefname{section}{Sec.}{Sections}
\else
\fi

\ifnum\llncs=1
\renewcommand*{\backref}[1]{}
\else
	\renewcommand*{\backref}[1]{(Cited on page~#1.)}
	\ifnum\notxfont=1
	\else
		\usepackage{newtxtext}
	\fi
\fi

\newcommand{\taiga}[1]{$\ll$\textsf{\color{red} Taiga: { #1}}$\gg$}
\newcommand{\mor}[1]{$\ll$\textsf{\color{blue} Tomoyuki: { #1}}$\gg$}

\newcommand{\Samp}{\algo{Samp}}
\newcommand{\Ver}{\algo{Ver}}
\newcommand{\ans}{\mathsf{ans}}
\newcommand{\puzz}{\mathsf{puzz}}

\newcommand{\la}{\leftarrow}
\newcommand{\ra}{\rightarrow}

\newcommand{\seteq}{\coloneqq}

\newcommand{\cA}{\mathcal{A}}
\newcommand{\cB}{\mathcal{B}}

\newcommand{\cD}{\mathcal{D}}

\newcommand{\cL}{\mathcal{L}}
\newcommand{\cM}{\mathcal{M}}

\newcommand{\cP}{\mathcal{P}}
\newcommand{\cQ}{\mathcal{Q}}
\newcommand{\cR}{\mathcal{R}}
\newcommand{\cS}{\mathcal{S}}

\newcommand{\cV}{\mathcal{V}}

\def\makeuppercase#1{
\expandafter\newcommand\csname tl#1\endcsname{\widetilde{#1}}
}

\def\makelowercase#1{
\expandafter\newcommand\csname tl#1\endcsname{\widetilde{#1}}
}

\newcommand{\N}{\mathbb{N}}

\newcommand{\R}{\mathbb{R}}

\newcommand*{\algo}[1]{\ensuremath{\mathsf{#1}}}

\newenvironment{boxfig}[2]{\begin{figure}[#1]\fbox{\begin{minipage}{0.97\linewidth}
                        \vspace{0.2em}
                        \makebox[0.025\linewidth]{}
                        \begin{minipage}{0.95\linewidth}
            {{
                        #2 }}
                        \end{minipage}
                        \vspace{0.2em}
                        \end{minipage}}}{\end{figure}}

\newcommand{\bit}{\{0,1\}}

\newcommand{\Gen}{\algo{Gen}}

\newcommand{\negl}{{\mathsf{negl}}}

\newcommand{\poly}{{\mathrm{poly}}}

\usepackage{tikz}

\ifnum\llncs=1
\let\oldvec\vec
\let\vec\oldvec
\makeatletter
\renewcommand*\l@author[2]{}
\renewcommand*\l@title[2]{}
\makeatletter

\else
\fi    

\theoremstyle{remark}

\newtheorem{fact}[theorem]{Fact}

\title{
\textbf{
Quantum Cryptography and Meta-Complexity
} 
}

\begin{document}

\ifnum\anonymous=1
\author{\empty}
\else
%
%
\ifnum\llncs=1
\index{Taiga, Hiroka}
\author{
	Taiga Hiroka\inst{1} 
}
\institute{
	Yukawa Institute for Theoretical Physics, Kyoto University, Japan  \and NTT Corporation, Tokyo, Japan
}
\else
%
%

\author[1]{Taiga Hiroka}
\author[1]{Tomoyuki Morimae}
\affil[1]{{\small Yukawa Institute for Theoretical Physics, Kyoto University, Kyoto, Japan}\authorcr{\small \{taiga.hiroka,tomoyuki.morimae\}@yukawa.kyoto-u.ac.jp}}

\renewcommand\Authands{, }
\fi 
\fi

\ifnum\llncs=1
\date{}
\else
\date{}
\fi

\maketitle


\pagenumbering{gobble} 

\begin{abstract}
In classical cryptography, one-way functions (OWFs) are the minimal assumption, while
it is not the case in quantum cryptography.
Several new primitives have been introduced such as pseudorandom unitaries
(PRUs), pseudorandom function-like state generators (PRFSGs), pseudorandom state
generators (PRSGs), one-way state generators (OWSGs), one-way puzzles (OWPuzzs), and
EFI pairs. They seem to be weaker than OWFs, but
still imply many useful applications such as private-key quantum money schemes, secret-key encryption, message authentication codes, digital
signatures, commitments, and multiparty computations.
Now that the possibility of quantum cryptography without OWFs has opened up, the most important
goal in the field is to build a foundation of it. 
In this paper, we, for the first time, characterize quantum cryptographic primitives with meta-complexity. 
We show that one-way puzzles (OWPuzzs) exist if and only if
GapK is weakly-quantum-average-hard.
GapK is a promise problem to decide whether a given bit string has a small Kolmogorov complexity or not.
Weakly-quantum-average-hard means that an instance is sampled from a QPT samplable distribution,
and for any QPT adversary the probability that it makes mistake is larger than ${\rm 1/poly}$.
We also show that if quantum PRGs exist then GapK is strongly-quantum-average-hard.
Here, strongly-quantum-average-hard is a stronger version of weakly-quantum-average-hard
where the probability that the adversary makes mistake is larger than $1/2-1/{\rm poly}$.
Finally, we show that if GapK is weakly-classical-average-hard, then inefficient-verifier proofs of quantumness (IV-PoQ) exist.
Weakly-classical-average-hard is the same as weakly-quantum-average-hard except that the adversary is PPT. 
IV-PoQ are a generalization of proofs of quantumness (PoQ) that capture sampling-based and search-based quantum advantage,
and an important application of OWpuzzs.
This is the fist time that quantum advantage is based on meta-complexity.
(Note: There are two concurrent works, \cite{cryptoeprint:2024/1490,cavalarmeta}.)
\end{abstract}

\ifnum\cameraready=1
\else
\ifnum\llncs=1
\else
\newpage
  \setcounter{tocdepth}{2}      
  \setcounter{secnumdepth}{2}   
  \tableofcontents
  \pagenumbering{arabic}
  \setcounter{page}{0}          
  \thispagestyle{empty}
  \clearpage
\fi
\fi

\section{Introduction}

In classical cryptography, the existence of one-way functions (OWFs) is the minimal assumption~\cite{FOCS:ImpLub89}, because
they are existentially equivalent to many primitives, such as pseudorandom generators (PRGs)~\cite{Hill99}, 
pseudorandom functions (PRFs)~\cite{JACM:GolGolMic86}, 
secret-key encryption (SKE)~\cite{GM84}, message authentication codes (MAC)~\cite{C:GolGolMic84}, digital signatures~\cite{STOC:Rompel90}, and commitments~\cite{C:Naor89}. 
Moreover, almost all primitives (including public-key encryption and multiparty computations) imply OWFs.

On the other hand, in quantum cryptography,
OWFs are not necessarily the minimum assumption~\cite{TQC:Kre21,C:MorYam22,C:AnaQiaYue22}.
Several new primitives have been introduced such as pseudorandom unitaries (PRUs)~\cite{C:JiLiuSon18}, 
pseudorandom function-like state generators (PRFSGs)~\cite{C:AnaQiaYue22,TCC:AGQY22},
pseudorandom state generators~(PRSGs)~\cite{C:JiLiuSon18}, one-way state generators (OWSGs)~\cite{C:MorYam22}, 
one-way puzzles~(OWPuzzs)~\cite{STOC:KhuTom24}, and EFI pairs~\cite{ITCS:BCQ23}. 
Although they could be weaker than OWFs~\cite{TQC:Kre21,STOC:KQST23,STOC:LomMaWri24},
they still imply many useful applications such as private-key quantum money schemes~\cite{C:JiLiuSon18}, SKE~\cite{C:AnaQiaYue22}, 
MAC~\cite{C:AnaQiaYue22}, 
digital signatures~\cite{C:MorYam22}, commitments~\cite{C:MorYam22,C:AnaQiaYue22,AC:Yan22}, 
and multiparty computations~\cite{C:MorYam22,C:AnaQiaYue22,C:BCKM21b,EC:GLSV21}.

Now that the possibility of the ``OWFs-free'' quantum cryptographic world (so-called Microcrypt) has opened up,
the most important goal in the field is to build a foundation of Microcrypt.
In classical cryptography, OWFs are founded in several ways and levels.
Although basing OWFs on $\mathbf{P}\neq\mathbf{NP}$ or
its average version is still open, OWFs have many instantiations based on concrete hardness assumptions, such as
the hardness of discrete logarithm~\cite{DifHel76} or lattice problems~\cite{STOC:Ajtai96},
and abstracted assumptions such as cryptographic group actions~\cite{cryptoeprint:2006/291,TCC:JQSY19,AC:ADMP20,C:BraYun90}.
Moreover, recent active studies have succeeded to base OWFs on meta-complexity (for example~\cite{FOCS:LiuPas20,IRS21}).

On the other hand, for the foundation of Microcrypt, currently what we know is only that
all Microcrypt primitives can be generically constructed at least from OWFs~\cite{cryptoeprint:2024/1652,C:JiLiuSon18,C:AnaQiaYue22,C:MorYam22,TQC:MorYam24,STOC:KhuTom24}. 
In particular, we do not know whether they can be based on some worst-case or average-case complexity assumptions,
hardness assumptions of concrete mathematical problems,
or some meta-complexity assumptions that do not imply OWFs.\footnote{There are two concurrent works~\cite{cryptoeprint:2024/1490,cavalarmeta} that tackle this open problem. See \cref{sec:concurrent_work}.}

\subsection{Our Results}
The goal of the present paper is to base Microcrypt primitives on meta-complexity.
As far as we know, this is the first time that quantum cryptographic primitives are characterized by meta-complexity.\footnote{There is a concurrent work \cite{cavalarmeta}. See \cref{sec:concurrent_work}.}
In the following, we explain each result.
Our results are also summarized in \cref{fig:graph}.

\paragraph{One-way puzzles.}
Our first result is a characterization of OWPuzzs with average-hardness of GapK.
\begin{theorem}\label{thm:main}
OWPuzzs exist if and only if GapK is weakly-quantum-average-hard. 
\end{theorem}
OWPuzzs~\cite{STOC:KhuTom24} are a quantum analogue of OWFs, and one of the most fundamental primitives in quantum cryptography.
A OWPuzz is a pair $(\Samp,\Ver)$ of two algorithms. $\Samp$ is a quantum polynomial-time (QPT) algorithm that, on input the security parameter $1^n$, outputs
two classical bit strings, $\ans$ and $\puzz$. $\Ver$ is an unbounded algorithm that, on input $(\puzz,\ans')$, outputs $\top/\bot$.
Correctness requires that $\Ver$ accepts $(\puzz,\ans)\gets\Samp(1^n)$ with high probability.
Security requires that no QPT adversary that receives $\puzz$ can output $\ans'$ such that $(\puzz,\ans')$ is accepted by $\Ver$ with high probability.
OWPuzzs are implied by almost all primitives including PRUs, PRFSGs, PRSGs, OWSGs, SKE, MAC, private-key quantum money schemes, etc.~\cite{TQC:MorYam24,C:MorYam22,STOC:KhuTom24,BatJai24}.
OWPuzzs imply EFI pairs, commitments, multiparty computations, and quantum advantage~\cite{STOC:KhuTom24,EC:GLSV21,C:BCKM21b,MorShiYam24}.

GapK~\cite{IRS21} is a promise problem 
to decide whether a given classical bit string $x$ has a small Kolmogorov complexity or not.
Roughly speaking, a Kolmogorov complexity~\cite{Sol64,Kol68,Cha69} of a bit string $x$ is the length of 
the shortest program that outputs $x$. (For more details, see for example \cite{LV19}.)
Weakly-quantum-average-hard means that
an instance $x$ is sampled from a QPT samplable distribution,
and for any QPT adversary the probability that it makes mistake is larger than $1/\poly$.
(Here, the probability is taken over the sampling of the instance and the algorithm of the adversary.)

\paragraph{Quantum PRGs.}
Having characterized OWPuzzs, the next question is whether we can do the same for other Microcrypt primitives.
In paricular, OWPuzzs are a search-type primitive, and hence it is interesting to ask whether we can characterize decision-type primitives with meta-complexity.
Our second result is a meta-complexity lower-bound for quantum PRGs (QPRGs).
\begin{theorem}
If QPRGs exist, then GapK is strongly-quantum-average-hard.    
\end{theorem}
A QPRG is a QPT algorithm that takes the security parameter $1^n$ as input and outputs a classical bit string.
Its output probability distribution is statistically far but computationally indistinguishable from a uniform distribution.
Strongly-quantum-average-hard is a stronger version of weakly-quantum-average-hard where
the probability that a QPT adversary makes mistake is larger than $1/2-1/\poly$.
QPRGs are a special case of QEFID and EFI. An EFI~\cite{ITCS:BCQ23} is a QPT algorithm that outputs two quantum states that are statistically far
but computationally indistinguishable. QEFID are a variant of EFI where the outputs are bit strings.
We left the problem of characterizing EFI and QEFID with meta-complexity open.
Another open problem is whether we can establish a meta-complexity upperbound for QPRGs.

\paragraph{Quantum advantage.}
Quantum advantage is another important application of OWPuzzs~\cite{MorShiYam24}.
We can base quantum advantage on meta-complexity.
As far as we know, this is the first time that quantum advantage is based on meta-complexity.
\begin{theorem}\label{thm:IV_POQ}
If GapK is weakly-classical-average-hard,    
then IV-PoQ exist.
\end{theorem}
Weakly-classical-average-hard is equivalent to weakly-quantum-average-hard except that
the adversary is PPT. (Note that the instance sampling algorithm is still QPT.)
Inefficient-verifier proofs of quantumness (IV-PoQ)~\cite{C:MorYam24} 
are a generalization of proofs of quantumness (PoQ)~\cite{JACM:BCMVV21} and capture various notions of quantum advantage such as sampling-based 
quantum advantage~\cite{NatPhys:BFNV19,STOC:AarArk11,TD04,BreJozShe10,BreMonShe16,FKMNTT18}
and searching-based one~\cite{CCC:AarChe17,AarGun19,STOC:Aaronson10,STOC:ACCGSW23,FOCS:Shor94,10.1145/3658665}.
An IV-PoQ is an interactive protocol over a classical channel between a verifier and a QPT prover.
During the interaction phase, the verifier is PPT, but after the interaction, the verifier becomes unbounded.
If the QPT prover honestly runs the protocol, the unbounded verifier accepts with high probability, but
for any PPT prover, the unbounded verifier does not accept except for a small probability.
It was shown recently~\cite{MorShiYam24} that IV-PoQ are existentially equivalent to
classically-secure OWPuzzs. 
This in particular means that
OWPuzzs imply IV-PoQ.
Our result, for the first time, bases quantum advantage on meta-complexity.
It is an open problem whether any meta-complexity lower bound for quantum advantage can be established.

\paragraph{Quantum probability estimation.}
\cref{thm:main} shows that average-hardness of GapK implies OWPuzzs.
However, technically, we do not directly construct OWPuzzs from the hardness of GapK.
We introduce a new problem, which we call quantum probability estimation (QPE),
and show the following two results from which GapK$\Rightarrow$OWPuzzs is obtained.
\begin{theorem}
If GapK is weakly-quantum-average-hard, then QPE is quantum-average-hard.
\end{theorem}
\begin{theorem}
If QPE is quantum-average-hard, then OWPuzzs exist.
\end{theorem}

QPE is, roughly speaking, the task of computing 
$\Pr[x\gets \cQ(1^n)]$ 
within a multiplicative-error 
given $x$ and a classical description of a QPT algorithm $\cQ$ that outputs $n$-bit strings on input $1^n$.
Its quantum-average-hardness
means that an instance $x\in\bit^n$ is sampled by $\cQ(1^n)$,
and for any QPT adversary the probability that it outputs a correct estimate
is less than $1-1/\poly(n)$.

The problem QPE, and the assumption of its average-hardess
themselves seem to be of independent interest, and will be useful in other applications.\footnote{In fact, in a concurrent paper \cite{cryptoeprint:2024/1490},
Khurana and Tomer showed that hardness of QPE is implied by the assumption previously studied in the field of sampling-based
quantum advantage (plus $\mathbf{P}^{\#\mathbf{P}}\not\subseteq(io)\mathbf{BQP/qpoly}$).}

\usetikzlibrary{positioning} 
\usetikzlibrary{calc} 
\usetikzlibrary {quotes}
\tikzset{>=latex} 

\tikzstyle{mysmallarrow}=[->,black,line width=1.6]
\tikzstyle{mydbarrow}=[<->,black,line width=1.6]
\tikzstyle{newarrow}=[<->,red,line width=1.6]
\tikzstyle{newsinglearrow}=[->,red,line width=1.6]
\tikzstyle{carrow}=[<->,red,line width=1.6]
\begin{figure}
\begin{center}
    \begin{tikzpicture}[scale=0.9,every edge quotes/.style = {font=\footnotesize,fill=white}]
      \def\h{-2.0} 
      \def\w{2.6} 

        \node[] (sharpP) at (2.9*\w,1*\h) {$\mathbf{P}^{\# \mathbf{P}}\not\subseteq \mathbf{(io)BQP/qpoly}$};
        \node[] (PH) at (0*\w,1.6*\h) {$\mathbf{P}^{\# \mathbf{P}}\not\subseteq \mathbf{(io)BPP}^{\bf NP}$};
        \node[] (dummy2) at (0.8*\w,2.5*\h) {};
        \node[] (QAA) at (0.8*\w,0.5*\h) {Quantum advantage assumption};
        \node[] (QPE) at (2*\w,3*\h) {QPE};
        \node[] (dummy) at (2*\w,2*\h) {};
       \node[] (QPRGs) at (4.2*\w,0.5*\h) {QPRGs};
       \node[] (QEFID) at (4.2*\w,1.5*\h) {QEFID};
       \node[] (sGapK) at (5.3*\w,1.0*\h) {sGapK};
       \node[] (OWPuzzs) at (4.2*\w,3*\h) {OWPuzzs};
       \node[] (nuQPRGs) at (5.5*\w,4*\h) {nuQPRGs};
        \node[] (IV-PoQ) at (2.5*\w,5*\h) {IV-PoQ};
        \node[] (QASs) at (1.5*\w,4*\h) {QASs};
        \node[] (Samp) at (0.8*\w,5*\h) {$\mathbf{SampBPP}\neq\mathbf{SampBQP}$};
        \node[] (OWForSamp) at (2.5*\w,6*\h) {qs-OWFs or $\mathbf{SampBPP}\neq\mathbf{SampBQP}$};
        \node[] (wGapK) at (6*\w,1.5*\h) {wGapK};
        \node[] (csOWPuzzs) at (4.2*\w,5*\h) {cs-OWPuzzs};
        \node[] (cwGapK) at (6.3*\w,4.3*\h) {cwGapK};
        \node[] (QAS/OWF) at (6*\w,5*\h) {QAS/OWF};


        \draw[mysmallarrow=black] (sharpP) edge[] (dummy);
        \draw[mysmallarrow=black] (QAA) edge[] (dummy);
        \draw[mysmallarrow=black] (dummy) edge["\cite{cryptoeprint:2024/1490}"] (QPE);
        \draw[newsinglearrow=red] (QPE) edge[] (OWPuzzs);
        \draw[mysmallarrow=black] (OWPuzzs) edge[] (csOWPuzzs);
        \draw[newsinglearrow=red] (wGapK) edge[] (QPE);
        \draw[newsinglearrow=red] (cwGapK) edge[] (QAS/OWF);
        \draw[mydbarrow=black] (IV-PoQ) edge["\cite{MorShiYam24}"] (csOWPuzzs);
        \draw[mysmallarrow=black] (IV-PoQ) edge["\cite{MorShiYam24}"] (OWForSamp);
        \draw[mysmallarrow=black] (QASs) edge[] (IV-PoQ);
        \draw[mysmallarrow=black] (QASs) edge[] (Samp);
        \draw[mysmallarrow=black] (dummy2) edge["\cite{STOC:AarArk11,BreMonShe16}"] (Samp);
        \draw[mysmallarrow=black] (dummy2) edge["\cite{MorShiYam24}"] (QASs);
        \draw[mysmallarrow=black] (QAA) edge[] (dummy2);
        \draw[mysmallarrow=black] (PH) edge[] (dummy2);
        \draw[mysmallarrow=black] (QPRGs) edge[] (QEFID);
        \draw[mysmallarrow=black] (QEFID) edge[] (OWPuzzs);
        \draw[mysmallarrow=black] (sGapK) edge[] (wGapK);
        \draw[mysmallarrow=black] (wGapK) edge[] (cwGapK);
        \draw[newsinglearrow=red] (QPRGs) edge[] (sGapK);
        \draw[mydbarrow=black] (nuQPRGs) edge["\cite{STOC:KhuTom24,C:ChuGolGra24}"] (OWPuzzs);
        \draw[mydbarrow=black] (QAS/OWF) edge["\cite{MorShiYam24}"] (csOWPuzzs);
        \draw[newsinglearrow=red] (nuQPRGs) edge[] (wGapK);
    \end{tikzpicture}
\end{center}
\caption{A summary of results. Black lines are known results or trivial implications.
Red lines are new in our work. 
The concurrent work \cite{cryptoeprint:2024/1490} showed QPE$\Leftrightarrow$OWPuzzs.
The concurrent work \cite{cavalarmeta} showed OWPuzzs$\Rightarrow$wGapK, wGapK$\Rightarrow$QPE, and QPE$\Rightarrow$OWPuzzs.
``qs'' stands for quantumly-secure, and ``cs'' stands for classically-secure.
``nu'' stands for non-uniform.
sGapK means strongly-quantum-average-hardness of GapK.
wGapK means weakly-quantum-average-hardness of GapK.
cwGapK means weakly-classical-average-hardness of GapK.
Quantum advantage assumption is the combination of so-called ``$\#\mathbf{P}$-hardness assumption'' (computing some functions within multiplicative errors
is $\#\mathbf{P}$-hard on average) and so-called ``anti-concentration''~\cite{STOC:AarArk11,BreMonShe16}.)
}
\label{fig:graph}
\end{figure}

\subsection{Technical Overview}
In this subsection, we provide high-level overview of our proofs. 

\paragraph{OWPuzzs from weakly-quantum-average-hardness of GapK.}
$\mathsf{GapK}[s_1,s_2]$ is the following promise problem:
Given a bit string $x$, decide $K(x)\le s_1$ or $K(x)\ge s_2$,
where $K(x)$ is the Kolmogorov complexity of $x$, i.e., the length of the shortest program that a universal Turing machine outputs $x$.
Weakly-quantum-average-hardness means that an instance $x$ is sampled from a QPT samplable distribution, and
for any QPT adversary, the probability that it makes mistake is larger than $1/\poly$.
(Here, the probability is taken over the sampling of the instance and the algorithm of the adversary.)
More precisely, we require that there exist an integer $k>0$ and a QPT algorithm $\cQ$ 
(that takes $1^n$ as input and outputs $n$-bit strings)
such that, for any QPT adversary $\cA$,
  \begin{align}
        \Pr_{x\la\cQ(1^n)}[ \mathsf{no}\la \cA(x) \wedge x\in\cL_{\mathsf{Yes}} ] +\Pr_{x\la\cQ(1^n)}[\mathsf{yes}\la\cA(x)\wedge x\in\cL_{\mathsf{No}} ]\geq \frac{1}{n^k}
    \end{align}
    for all sufficiently large $n\in\N$, 
    where $\cL_{\mathsf{Yes}}$ (resp. $\cL_{\mathsf{No}}$) is the set of yes (resp. no) instances of $\mathsf{GapK}[s_1,s_2]$.\footnote{A bit string $x$ is a yes (no) instance if $K(x)\le s_1$ ($K(x)\ge s_2$).}
    

We do not directly show the existence of OWPuzzs from weakly-quantum-average-hardness of GapK. 
We introduce another problem, which we call quantum probability estimation (QPE), and
assume its quantum-average-hardness.
We first show that weakly-quantum-average-hardness of GapK implies quantum-average-hardness of QPE, and then
we show that quantum-average-hardness of QPE implies OWPuzzs.

A QPE is a problem of computing $\Pr[x\gets\cQ(1^n)]$ within a multiplicative error given $x$ and a classical description of
a QPT algorithm $\cQ$.
Its quantum-average-hardness means that $x$ is sampled from $\cQ(1^n)$ and for any QPT adversary,
the probability that it outputs a correct estimate is smaller than $1-1/\poly(n)$.
More precisely, 
we require that there exist a real $c>1$, an integer $q>0$, and a QPT algorithm $\cQ$
(that takes $1^n$ as input and outputs $n$-bit strings)
such that, 
for any QPT algorithm $\mathsf{Estimate}$, 
    \begin{align}
        \Pr_{x\la\cQ(1^n)}\left[\frac{1}{c}\Pr[x\la\cQ(1^n)]\leq \mathsf{Estimate}(x)\leq c \Pr[x\la\cQ(1^n)]\right]\leq 1-\frac{1}{n^{q}}
    \end{align}
for all sufficiently large $n\in\N$.

Let us first explain how to construct OWPuzzs from quantum-average-hardness of QPE.
Let $\cQ$ be the instance sampling QPT algorithm for quantum-average-hardness of QPE. 
\cite{C:ChuGolGra24} showed the equivalence between OWPuzzs and distributional OWPuzzs,
which means that if OWPuzzs do not exist, then distributional OWPuzzs do not exist as well.
If distributional OWPuzzs do not exist, then
there exists a QPT extrapolation algorithm $\mathsf{Ext}$ that, given $(x_1,...,x_{i-1})\gets \cQ(1^n)$, can sample $x_i$, such that
the statistical distance between $(x_1,...,x_i)_{(x_1,...,x_i)\gets\cQ(1^n)}$ 
and $(x_1,...,x_{i-1},\mathsf{Ext}(x_1,...,x_{i-1}))_{(x_1,...,x_{i-1})\gets\cQ(1^n)}$
is small for all $i\in[n-1]$. Then, by repeatedly running $\mathsf{Ext}$, we can compute an approximation of $\Pr[(x_1,...,x_n)\gets\cQ(1^n)]$.
This means that quantum-average-hardness of QPE does not hold.

Next, let us explain how quantum-average-hardness of QPE is derived
from weakly-quantum-average-hardness of GapK.\footnote{This was essentially shown in \cite{IRS21}, but we slightly improved the parameters for our purpose.
}
Assume that
quantum-average-hardness of QPE does not hold.
Then, there is a QPT algorithm $\mathsf{Estimate}$ that can estimate output probability of any QPT distribution.
We construct a QPT algorithm $\cA$ that solves GapK as follows: 
Let $\cQ$ be a QPT algorithm that samples the instance of GapK.
Given $x\gets\cQ(1^n)$, run $\mathsf{Estimate}(x)$ and output yes if the approximation of $\Pr[x\gets \cQ(1^n)]$
is larger than a certain threshold.
Intuitively, $K(x)$ is small (resp. large) if $\Pr[x\gets\cQ(1^n)]$ is large (resp. small) with high probability over $x\la\cQ(1^n)$\footnote{If $\Pr[x\gets\cQ(1^n)]$ is large, $x$ can be taken from the support of $\cQ(1^n)$, which allows
a shorter program to output $x$. 
On the other hand, because there are at most $2^{s+1}$ strings $x\in\bit^n$ that satisfy $K(x)<s$, the probability that 
$\cQ(1^n)$ outputs $x$ such that both $K(x)<s$ and $\Pr[x\la\cQ(1^n)]<\alpha$ is at most $2^{s+1}\alpha$.
Therefore, if $\Pr[x\la\cQ(1^n)]$ is sufficiently small, then $K(x)$ must be large with high probability over $x\la\cQ(1^n)$. },
and therefore $\cA$ can correctly decide whether $K(x)$ is small or large.

\paragraph{Strongly-quantum-average-hardness of GapK from QPRGs.}
Next let us explain how to show strongly-quantum-average-hardness of GapK from the existence of QPRGs.
Strongly-quantum-average-hardness is the same as weakly-quantum-average-hardness explained above
except that the probability that the adversary makes a mistake is larger than $1/2-1/\poly$.
A QPRG is a QPT algorithm $\Gen$ that takes the security parameter $1^n$ as input 
and outputs $n$-bit strings whose distribution is statistically far but computationally indistinguishable from the uniform distribution
over $\bit^n$.
Our key observation is that if the output distribution of
$\Gen(1^n)$ is exponentially statistically far\footnote{We say that a family $\{D_n\}_{n\in\mathbb{N}}$ 
of distributions where $D_n$ is a distribution over $\bit^n$
is exponentially statistically far from a family $\{E_n\}_{n\in\mathbb{N}}$ of distributions where $E_n$ is a distribution over $\bit^n$, 
if for some real $0<\tau<1$, $\mathsf{SD}((x)_{x\la D_n}, (x)_{x\la E_n})>1-2^{-n^{\tau}}$ for all sufficiently large $n\in\mathbb{N}$. 
} from the uniform distribution over $\bit^n$, then 
the Kolmogorov complexity $K(x)$ of $x$ is small with high probability over $x\gets\Gen(1^n)$. 
In that case, it is easy to show strongly-quantum-average-hardness of GapK as follows.
Let $\cQ$ be a QPT algorithm that outputs $x\la \Gen(1^n)$ with probability $1/2$ and a uniformly random $n$-bit string $x$ with probability $1/2$.
If there is a QPT algorithm that
can decide if $K(x)$ is large or small with probability larger than $1/2+1/\poly(n)$ over the distribution $x\gets\cQ(1^n)$, 
then we can construct a QPT algorithm that can distinguish the distribution of $\Gen(1^n)$ from the uniform distribution over $\bit^n$,
which breaks the QPRG.

However, in general, a QPRG does not satisfy such an exponentially-statistically-far property. 
Fortunately, we can show that if QPRGs exist, then
QPRGs with the exponentially-statistically-far property exist
by the parallel repetition and the padding argument.

\paragraph{Weakly-quantum-average-hardness of GapK from OWPuzzs.}
This is obtained by showing that non-uniform QPRGs (nuQPRGs) imply weakly-quantum-average-hardness of GapK,
because nuQPRGs exist if and only if OWPuzzs exist~\cite{STOC:KhuTom24,C:ChuGolGra24}.

nuQPRGs are a non-uniform version of QPRGs:
A nuQPRG is a QPT algorithm $\Gen$ that takes the security parameter $1^n$ and an advice bit string $\mu$ of $\log(n)$ length as input,
and outputs $n$-bit strings.
It satisfies that for a $\mu^*\in [n]$, the distribution $\Gen(1^n,\mu^*)$ is statistically far but computationally indistinguishable from the uniform distribution
over $\bit^n$.

The proof of showing weakly-quantum-average-hardness of GapK from nuQPRGs is similar to that of
strongly-quantum-average-hardness of GapK from QPRGs.
As we have already explained above,
in the proof of showing strongly-quantum-average-hardness of GapK from QPRGs,
we define a QPT algorithm $\cQ$ that outputs $x\gets\Gen(1^n)$ with 1/2 and outputs random $n$-bit string with probability 1/2,
and show that $\cQ$ works as an instance sampling algorithm for which GapK is strongly-quantum-average-hard.
In order to show 
weakly-quantum-average-hardness of GapK from nuQPRGs,
a difference is that $\cQ$ additionally samples an advice string $\mu\la [n]$ 
and outputs $x\la \Gen(1^n,\mu)$ with probability $1/2$ and a random $n$-bit string with probability $1/2$.
Because the good advice $\mu^*$ is obtained with probability $1/n$, weakly-quantum-average-hardness of GapK is satisfied.

\if
We consider the distribution $\cD$ that randomly samples $\mu^*$ and outputs $x\la G(\mu^*)$ with probability $1/2$ and outputs uniform random string $x$ with probability $1/2$.
If we can check if $K(x)$ is high or small, then we can break the security of non-uniform QPRG $G$ because with inverse polynomial probability $G(\mu^*)$ is statistically far from uniform distribution, and hence $K(x)$ is low over the distribution $G(\mu^*)$.
\fi

\paragraph{IV-PoQ from weakly-classical-average-hardness of GapK.} 
From the result of \cite{MorShiYam24}, IV-PoQ exist if and only if the QAS/OWF condition holds. Therefore we show
that
if weakly-classical-average-hardness of GapK is satisfied, then the QAS/OWF condition is satisfied.

Informally, if the QAS/OWF condition is satisfied, then quantum advantage samplers (QASs) exist or OWFs exist.
Here a QAS is QPT algorithm that outputs classical bit strings whose distribution
cannot be sampled with any PPT algorithm.
For contradiction, assume that the QAS/OWF condition is not satisfied,
which roughly means that both QASs and OWFs do not exist.
Let $\cQ$ be a QPT algorithm that samples the instance $x$ of GapK. Because QASs do not exist,
there exists a PPT algorithm $\cS$ that approximately samples $\cQ$.
Moreover, because OWFs do not exist, we can construct a PPT algorithm $\mathsf{Estimate}$ that estimates $\Pr[x\la\cS(1^n)]\approx\Pr[x\la\cQ(1^n)]$.
As we have explained above, if we can estimate the value of $\Pr[x\la\cQ(1^n)]$ then we can solve GapK.

As is shown in \cite{MorShiYam24}, classically-secure OWPuzzs exist if and only if IV-PoQ exist.
Therefore, readers might think that we can directly show that weakly-classical-average-hardness of GapK implies classically-secure OWPuzzs
(and hence IV-PoQ)
by replacing all QPT adversaries with PPT ones in the proof of OWPuzzs from weakly-quantum-average-hardness of GapK.
However, we do not know how to do that, because
we do not know whether classically-secure OWPuzzs are equivalent to classically-secure distributional OWPuzzs unlike the
quantumly-secure case~\cite{C:ChuGolGra24}.

\subsection{Concurrent Works}
\label{sec:concurrent_work}
There are two concurrent works, \cite{cryptoeprint:2024/1490,cavalarmeta}. We explain relations between
our results and these concurrent works.

\paragraph{Relation to \cite{cryptoeprint:2024/1490}.}
During the preparation of the first version of this manuscript, Khurana and Tomer uploaded
\cite{cryptoeprint:2024/1490}.
They introduce two hardness assumptions, Assumption 1 and Assumption 2.
Assumption 1 is an assumption often studied in the field of quantum advantage
to derive
sampling-based quantum advantage.
They show that Assumption 1 plus a mild complexity assumption,
$\mathbf{P}^{\# \mathsf{P}}\not\subseteq \mathbf{(io)BQP}/\mathbf{qpoly}$,
imply Assumption 2.
They then show that Assumption 2 is equivalent to the existence of OWPuzzs.

Their Assumption 2 is the same as our assumption, \cref{ass:AH_hardness_Estimate},
the average-hardness of QPE.
Moreover, their proof of ``Assumption 2 $\Rightarrow$ OWPuzzs'' is similar to our proof of ``\cref{ass:AH_hardness_Estimate} $\Rightarrow$ OWPuzzs''.

In addition, their Assumption 2 allows quantum advice for adversaries. In this paper we do not explicitly consider quantum advice, but we believe
that our proofs can be straightforwardly extended to the case with quantum advice.

\paragraph{Relation to \cite{cavalarmeta}.}
After uploading the first version of this paper on arXiv, Cavalar, Goldin, Gray, and Hall sent us their manuscript (that was later uploaded \cite{cavalarmeta}).
In the first version of our paper, we showed only the direction of wGapK$\Rightarrow$OWPuzzs, while their manuscript
shows the both directions wGapK$\Leftrightarrow$OWPuzzs.
When they contacted us, we had shown QPRGs$\Rightarrow$sGapK, and planned to add the result to the revised version.
We later realized that with almost the same proof, this also means nuQPRGs$\Rightarrow$wGapK, which means
OWPuzzs$\Rightarrow$wGapK. Our proof for OWPuzzs$\Rightarrow$wGapK is different from theirs.

\section{Preliminaries}
\subsection{Basic Notations}
We use the standard notations of cryptography and quantum information.
$n$ is the security parameter.
$\negl$ is a negligible function.
$[n]$ denotes the set $\{1,2,...,n\}$.
QPT stands for quantum polynomial time,
and PPT stands for classical polynomial time.
For an algorithm (or a Turing machine) $\cA$, $y\gets\cA(x)$ means that $\cA$ runs on input $x$ and outputs $y$.
$\mathsf{SD}(D,E)$ is the statistical distance between two distributions $D$ and $E$.

\subsection{One-Way Puzzles}
We review the definition of OWPuzzs.
\begin{definition}[One-Way Puzzles (OWPuzzs) \cite{STOC:KhuTom24}]
\label{def:OWPuzz}
A one-way puzzle is a pair $(\Samp, \Ver)$ of algorithms with the following syntax:
\begin{itemize}
    \item $\Samp(1^n) \rightarrow (\ans,\puzz)$: It is a QPT algorithm that, on input $1^n$, outputs two classical bit strings $(\ans,\puzz)$. 
    \item $\Ver(\ans',\puzz) \rightarrow \top/\bot$: It is an unbounded algorithm that, on input $(\ans',\puzz)$, outputs $\top/\bot$.
\end{itemize}
We require the following correctness and security.

\paragraph{Correctness:}
\begin{align}
\Pr_{(\ans,\puzz)\leftarrow \Samp(1^n)} [\top\gets\Ver(\ans,\puzz)] \ge 1 - \negl(n).
\end{align}

\paragraph{Security:}
For any \textbf{uniform} QPT adversary $\cA$,
\begin{align}
\Pr_{(\ans,\puzz) \leftarrow \Samp(1^n)} [\top\gets\Ver(\cA(1^n,\puzz),\puzz)] \le \negl(n).
\end{align}
\end{definition}

\cite{C:ChuGolGra24} introduced distributional OWPuzzs, and show their equivalence to OWPuzzs.
The following lemma comes from the equivalence. We will use the lemma later.
\begin{lemma}[\cite{C:ChuGolGra24}]\label{thm:dist_invert}
Suppose that (resp. infinitely-often) OWPuzzs do not exist.
Then, for any QPT algorithm $\cQ$ that on input $1^n$ outputs $n$-bit strings, and for any constant 
$t>0$,
there exists a uniform QPT algorithm $\mathsf{Ext}$ that outputs a single bit such that 
\begin{align}
    \mathsf{SD}\left( (y_1,...,y_i)_{(y_1,...,y_i)\gets\cQ(1^n)} , 
    (y_1,...,y_{i-1}, \mathsf{Ext}(1^n,i,y_1,..., y_{i-1}) )_{(y_1,...,y_{i-1})\gets \cQ(1^n)} \right)\leq \frac{1}{n^t} 
\end{align}
for all $i\in[n]$
and infinitely many $n\in\N$ (resp. all sufficiently large $n\in\N$).
Here $y_i$ is the $i$th bit of $y\in\bit^n$, and $(y_1,...,y_i)\gets\cQ(1^n)$ is the marginal distribution over
the first $i$ bits of the output of $\cQ(1^n)$.
\end{lemma}

\subsection{Quantum PRGs}

We review the definitions of quantum PRGs (QPRGs).
\begin{definition}[Quantum PRGs (QPRGs)]\label{def:uniform_qefid}
A QPRG is a QPT algorithm $\Gen(1^n)$ that takes a security parameter $1^n$ as input, and outputs a classical bit
string $x\in\bit^n$ satisfying the following properties:
    \paragraph{Statistically far:} 
    \begin{align}
        \mathsf{SD}((x)_{x\la\Gen(1^n)},(x)_{x\la \bit^n})\geq 1-\negl(n).
    \end{align}
    \paragraph{Computationally indistinguishable:} 
    For any QPT algorithm $\cA$,
    \begin{align}
        \abs{\Pr_{x\la\Gen(1^n)}[1\la\cA(1^n,x)]-\Pr_{x\la \bit^n}[1\la\cA(1^n,x)]}\leq\negl(n).
    \end{align}
\end{definition}

We will later use the following lemma.
Its proof is given by the standard padding argument. We give its proof in \cref{sec:appendix}.

\begin{lemma}\label{lem:far}
    Suppose that there exists a QPRG $\Gen$.
    Then, for any real $0<\tau<1$, there exists a QPRG $\Gen^*$ that satisfies the following properties.
        \paragraph{Exponentially statistically far:}
        \begin{align}
        \mathsf{SD}((x)_{x\la\Gen^*(1^n)},(x)_{x\la \bit^n})\geq 1-2^{-n^{\tau}}
    \end{align}
    for all sufficiently large $n\in\N$.
    \paragraph{Computationally indistinguishable:} 
    For any QPT algorithm $\cA$,
    \begin{align}
        \abs{\Pr_{x\la\Gen^*(1^n)}[1\la\cA(1^n,x)]-\Pr_{x\la \bit^n}[1\la\cA(1^n,x)]}\leq\negl(n).
    \end{align}
\end{lemma}

We also introduce a non-uniform version of QPRGs.
\begin{definition}[non-uniform QPRGs (nuQPRGs)] 
A non-uniform QPRG is a QPT algorithm $\Gen(1^n,\mu)$ that takes a security parameter $1^n$ and $\mu\in [n]$ as input, and outputs a classical 
bit string $x\in\bit^n$.
We require that there exists $\mu^*\in[n]$ such that the following two conditions hold:
    \paragraph{Statistically far:} 
    \begin{align}
        \mathsf{SD}((x)_{x\la\Gen(1^n,\mu^*)},(x)_{x\la \bit^n})\geq 1-\negl(n).
    \end{align}
    \paragraph{Computationally indistinguishable:} 
    For any QPT algorithm $\cA$,
    \begin{align}
        \abs{\Pr_{x\la\Gen(1^n,\mu^*)}[1\la\cA(1^n,x)]-\Pr_{x\la \bit^n}[1\la\cA(1^n,x)]}\leq\negl(n).
    \end{align}
\end{definition}

It is known that OWPuzzs are existentially equivalent to nuQPRGs.
\begin{theorem}[\cite{STOC:KhuTom24,C:ChuGolGra24}]\label{thm:QEFID_OWPuzz}
OWPuzzs exist if and only if nuQPRGs exist.    
\end{theorem}

We will use the following lemma.
We omit its proof because it is the same as that of \cref{lem:far}.
\begin{lemma}\label{lem:non_uniform_far}
    Suppose that there exists a nuQPRG $\Gen$.
    Then, for any real $0<\tau<1$, there exists another nuQPRG
    $\Gen^*$ such that there exists $\mu^*\in[n]$ with the following properties:
        \paragraph{Exponentially statistically far:}
        \begin{align}
        \mathsf{SD}((x)_{x\la\Gen^*(1^n,\mu^*)},(x)_{x\la \bit^n})\geq 1-2^{-n^{\tau}}
    \end{align}
    for all sufficiently large $n\in\mathbb{N}$.
    \paragraph{Computationally indistinguishable:}
    For any QPT algorithm $\cA$,
    \begin{align}
        \abs{\Pr_{x\la\Gen^*(1^n,\mu^*)}[1\la\cA(1^n,x)]-\Pr_{x\la \bit^n}[1\la\cA(1^n,x)]}\leq\negl(n).
    \end{align}
\end{lemma}

\subsection{One-Way Functions}
\begin{definition}[OWFs on $\Sigma$~\cite{MorShiYam24}]\label{def:OWFsSigma}
    Let $\Sigma\subseteq\mathbb{N}$ be a set.
    A function $f:\bit^*\to\bit^*$ 
    that is computable in classical deterministic polynomial-time
    is a classically-secure (resp. quantumly-secure) OWF on $\Sigma$ if
    there exists an efficiently-computable polynomial $m$ such that
    for any PPT (resp. QPT) adversary $\cA$ and any 
    polynomial $p$
    there exists $n^*\in\mathbb{N}$ such that
    \begin{equation}
    \Pr[f(x')=f(x): x\gets\bit^{m(n)}, x'\gets\cA(1^{m(n)},f(x))] \le\frac{1}{p(n)}
    \end{equation} 
    holds
    for all $n\ge n^*$ in $\Sigma$. 
\end{definition}

\subsection{QAS/OWF Condition}
\begin{definition}[The QAS/OWF Condition~\cite{MorShiYam24}]\label{def:QAS/OWF}
    The QAS/OWF condition holds if there exist a polynomial $p$, a QPT algorithm $\cQ$ that takes $1^n$ as input and outputs a classical string, 
    and a function $f:\bit^*\to\bit^*$ that is computable in classical deterministic polynomial-time
    such that for any PPT algorithm $\cS$, the following holds:
    if we define
    \begin{align}
        \Sigma_\cS := \left\{ n\in\mathbb{N} :  \mathsf{SD}(\cQ(1^n),\cS(1^n)) \le \frac{1}{p(n)} \right\},
    \end{align}
    then $f$ is a classically-secure OWF on $\Sigma_\cS$.
\end{definition}

We will use the following lemma.
\begin{lemma}[\cite{MorShiYam24}]\label{lem:shirakawa}
If the QAS/OWF condition is not satisfied, then the following statement
is satisfied:
for any QPT algorithm $\cQ$ that takes $1^{n}$ as input and outputs a classical string and for any real $k>0$, there exists a PPT algorithm $\cS$ such that for any efficiently-computable polynomial $m$ and any family $\{f_n:\bit^{m(n) }\ra \bit^*\}_{n\in\N}$ of functions that are computable in classical deterministic polynomial-time, there exists a PPT algorithm $\cR$ such that
\begin{align}
    \mathsf{SD}(\cQ(1^n),\cS(1^n))\leq \frac{1}{n^k}
\end{align}
and
\begin{align}
    \mathsf{SD}\left(\{x,f_n(x)\}_{x\la\bit^{m(n)}},\{\cR(1^{m(n)},f_n(x)),f_n(x)\}_{x\la\bit^{m(n)}} \right)\leq \frac{1}{n^k}
\end{align}
for infinitely many $n\in\N$.
\end{lemma}

\subsection{IV-PoQ}
\begin{definition}[Inefficient-Verifier Proofs of Quantumness (IV-PoQ)~\cite{C:MorYam24}]\label{def:IVPoQ}
    An IV-PoQ is a tuple $(\cP,\cV_1,\cV_2)$ of interactive algorithms. 
    $\cP$ (prover) is QPT, $\cV_1$ (first verifier) is PPT, and $\cV_2$ (second verifier) is unbounded.
    The protocol is divided into two phases.
    In the first phase, $\cP$ and $\cV_1$ take the security parameter $1^n$ as input and interact with each other over a classical channel.
    Let $\tau$ be the transcript, i.e., the sequence of all classical messages exchanged between $\cP$ and $\cV_1$.
    In the second phase, $\cV_2$ takes $1^n$ and $\tau$ as input and outputs $\top$ (accept) or $\bot$ (reject).
    We require the following two properties for some functions $c$ and $s$ such that $c(n)-s(n)\ge 1/\poly(n)$.
    \begin{itemize}
    \item
    {\bf $c$-completeness:} 
        \begin{equation}
            \Pr[\top\gets\cV_2(1^n,\tau):\tau\gets\langle\cP,\cV_1\rangle(1^n)] \ge c(n)
        \end{equation}
    holds for all sufficiently large $n\in\mathbb{N}$.
    \item
    {\bf $s$-soundness:} For any PPT prover $\cP^*$,
        \begin{equation}
            \Pr[\top\gets\cV_2(1^n,\tau):\tau\gets\langle\cP^*,\cV_1\rangle(1^n)] \le s(n)
        \end{equation}
    holds for all sufficiently large $n\in\mathbb{N}$.
    \end{itemize}
\end{definition}

\begin{theorem}[\cite{MorShiYam24}]
IV-PoQ exist if and only if the QAS/OWF condition is satisfied.    
\end{theorem}

\subsection{Kolmogorov Complexity}
We also review some basics of Kolmogorov complexity. For details, see for example \cite{LV19}.
Throughout this paper, we consider a fixed deterministic universal Turing machine $U$.

\begin{definition}[Kolmogorov Complexity]
    The Kolmogorov complexity $\mathrm{K}(x)$ for a string $x$ is defined as
    \begin{equation}
        \mathrm{K}(x) \coloneqq \min_{d\in\bit^*} \{|d| : x \la U(d) \}.
    \end{equation}
\end{definition}

\begin{definition}[$\mathsf{GapK}$]
Let $s_1:\N\ra\N$ and $s_2:\N\ra\N$ be functions such that $s_2(n)-s_1(n)>w(\log(n))$.
$\mathsf{GapK}[s_1,s_2]\coloneqq (\cL_{\mathsf{Yes}},\cL_{\mathsf{No}})\subseteq \bit^*$ is a promise problem whose yes instances are strings $x$ such that $K(x)\leq s_1(\abs{x})$ and no instances are strings $x$ such that $K(x)\geq s_2(\abs{x})$.
\end{definition}

\section{Assumptions}

In this section, we introduce assumptions.

\subsection{Average-Hardness of GapK}
We introduce three assumptions on average-hardness of GapK.
We first introduce strongly-quantum-average-hardness of GapK as follows.

\begin{assumption}[Strongly-Quantum-Average-Hardness of $\mathsf{GapK}\lbrack s_1,s_2\rbrack$ ]\label{ass:Strong_AH_hardness_GapK}
     Let $s_1:\N\ra\N$ and $s_2:\N\ra\N$ be polynomial-time-computable functions with $s_2(n)-s_1(n)>w(\log(n))$.
    There exist an integer $k>0$ and a QPT algorithm $\cQ$ that outputs $n$-bit strings on input $1^n$
    such that for any QPT algorithm $\cA$, 
    \begin{align}
        \Pr_{x\la\cQ(1^n)}[ \mathsf{no}\la \cA(x) \wedge x\in\cL_{\mathsf{Yes}}]
        +\Pr_{x\la\cQ(1^n)}[\mathsf{yes}\la\cA(x)\wedge x\in\cL_{\mathsf{No}} ]\geq \frac{1}{2}- \frac{1}{n^k}
    \end{align}
    for all sufficiently large $n\in\N$, where $\cL_{\mathsf{Yes}}$ (resp. $\cL_{\mathsf{No}}$) is the set of yes (resp. no) instances of $\mathsf{GapK}[s_1,s_2]$.
\end{assumption}

We next introduce a weaker version as follows where the failure probability is larger than $1/poly$, not $1/2-1/poly$.
\begin{assumption}[Weakly-Quantum-Average-Hardness of $\mathsf{GapK}\lbrack s_1,s_2 \rbrack$]\label{ass:AH_hardness_GapK}
    Let $s_1:\N\ra\N$ and $s_2:\N\ra\N$ be polynomial-time-computable functions with $s_2(n)-s_1(n)>w(\log(n))$.
    There exist an integer $k>1$ and a QPT algorithm $\cQ$ that outputs $n$-bit strings on input $1^n$
    such that for any QPT algorithm $\cA$, 
    \begin{align}
        \Pr_{x\la\cQ(1^n)}[ \mathsf{no}\la \cA(x) \wedge x\in\cL_{\mathsf{Yes}} ] 
        +\Pr_{x\la\cQ(1^n)}[\mathsf{yes}\la\cA(x)\wedge x\in\cL_{\mathsf{No}} ]\geq \frac{1}{n^k}
    \end{align}
    for all sufficiently large $n\in\N$, where $\cL_{\mathsf{Yes}}$ (resp. $\cL_{\mathsf{No}}$) is the set of yes (resp. no) instances of $\mathsf{GapK}[s_1,s_2]$.
\end{assumption}

We finally introduce a classical-average-hardness version which is equivalent to \cref{ass:AH_hardness_GapK} except that the adversary $\cA$ is PPT, not QPT.
(Note that the instance sampling algorithm is still QPT.)
\begin{assumption}[Weakly-Classical-Average-Hardness of $\mathsf{GapK}\lbrack s_1,s_2 \rbrack$]\label{ass:c_AH_hardness_GapK}
    Let $s_1:\N\ra\N$ and $s_2:\N\ra\N$ be polynomial-time-computable functions with $s_2(n)-s_1(n)\geq w(\log(n))$.
    There exist an integer $k>1$ and a QPT algorithm $\cQ$ that outputs $n$-bit strings on input $1^n$
    such that for any PPT algorithm $\cA$, 
    \begin{align}
        \Pr_{x\la\cQ(1^n)}[ \mathsf{no}\la \cA(x) \wedge x\in\cL_{\mathsf{Yes}} ] 
        +\Pr_{x\la\cQ(1^n)}[\mathsf{yes}\la\cA(x)\wedge x\in\cL_{\mathsf{No}} ]\geq \frac{1}{n^k}
    \end{align}
    for all sufficiently large $n\in\N$, where $\cL_{\mathsf{Yes}}$ (resp. $\cL_{\mathsf{No}}$) is the set of yes (resp. no) instances of $\mathsf{GapK}[s_1,s_2]$.
\end{assumption}

\subsection{Average-Hardness of QPE}
We introduce two types of average-hardness of QPE as follows.
\cref{ass:AH_hardness_Estimate} is defined against QPT adversaries, while
\cref{ass:c_AH_hardness_Estimate} is defined against PPT adversaries.

\begin{assumption}[Quantum-Average-Hardness of Quantum Probability Estimation]\label{ass:AH_hardness_Estimate}
    There exist a real $c>1$, an integer $q>0$, and a QPT algorithm $\cQ$ that outputs $n$-bit strings on input $1^n$
    such that, for any QPT algorithm $\mathsf{Estimate}$, we have
    \begin{align}
        \Pr_{x\la\cQ(1^n)}\left[\frac{1}{c}\Pr[x\la\cQ(1^n)]\leq \mathsf{Estimate}(x)\leq c \Pr[x\la\cQ(1^n)]\right]\leq 1-\frac{1}{n^{q}}
    \end{align}
    for all sufficiently large $n\in\N$. 
\end{assumption}

\begin{assumption}[Classical-Average-Hardness of Quantum Probability Estimation]\label{ass:c_AH_hardness_Estimate}
    There exist a real $c>1$, an integer $q>0$, and a QPT algorithm $\cQ$ that outputs $n$-bit strings on input $1^n$ such that, for any PPT algorithm $\mathsf{Estimate}$, we have
    \begin{align}
        \Pr_{x\la\cQ(1^n)}\left[\frac{1}{c}\Pr[x\la\cQ(1^n)]\leq \mathsf{Estimate}(x)\leq c \Pr[x\la\cQ(1^n)]\right]\leq 1-\frac{1}{n^{q}}
    \end{align}
    for all sufficiently large $n\in\N$.
\end{assumption}

\section{Results on OWPuzzs and QPRGs}

We show the following theorem. 

\begin{theorem}\label{thm:equivalence}
    The following three are equivalent:
    \begin{itemize}
        \item OWPuzzs exist.
        \item There exists a real $0<\epsilon<1$ and a polynomial-time-computable function $\Delta(n)=w(\log(n))$ such that \cref{ass:AH_hardness_GapK} holds with $s_1=n-n^{\epsilon}$ and $s_2=n-\Delta$.
        \item \cref{ass:AH_hardness_Estimate} holds.
    \end{itemize}
\end{theorem}
\cref{thm:equivalence} follows from the following three 
\cref{thm:owpuzz_from_kolmogorov,thm:estimate,thm:kolmogorov_from_owpuzz}.\footnote{\cref{thm:owpuzz_from_kolmogorov} was essentially shown in \cite{IRS21}. Here, we show a slightly stronger version for our purpose.}
 
\begin{lemma}[\cite{IRS21}]\label{thm:owpuzz_from_kolmogorov}
    Suppose that there exists a real $0<\epsilon<1$ and a polynomial-time-computable function $\Delta(n)=w(\log(n))$ such that \cref{ass:AH_hardness_GapK} holds with $s_1=n-n^{\epsilon}$ and $s_2=n-\Delta$.
    Then, \cref{ass:AH_hardness_Estimate} holds.
\end{lemma}

\begin{lemma}\label{thm:estimate}
       \cref{ass:AH_hardness_Estimate} implies the existence of OWPuzzs.
\end{lemma}

\begin{lemma}\label{thm:kolmogorov_from_owpuzz}
    If OWPuzzs exist, then for all real $0<\epsilon<1$, there exists a polynomial-time-computable function $\Delta(n)=w(\log(n))$ such that  \cref{ass:AH_hardness_GapK} holds with $s_1=n-n^{\epsilon}$ and $s_2=n-\Delta$.
\end{lemma}

\cref{thm:kolmogorov_from_owpuzz} is obtained by a slight modification of the following theorem.
\begin{theorem}\label{thm:QEFID_GapK}
    If QPRGs exist, then for all real $0<\epsilon<1$, there exists a polynomial-time-computable function $\Delta(n)=w(\log(n))$ such that  \cref{ass:Strong_AH_hardness_GapK} holds with $s_1=n-n^{\epsilon}$ and $s_2=n-\Delta$.
\end{theorem}

In the following subsections, we show them.

\subsection{Proof of \cref{thm:owpuzz_from_kolmogorov}}
In this subsection, we show \cref{thm:owpuzz_from_kolmogorov}.
The proof was essentially given in \cite{IRS21}, 
but here we re-provide a proof with stronger parameters, because it is convenient to show \cref{thm:estimate}
(and also for the convenience of readers).

\begin{proof}[Proof of \cref{thm:owpuzz_from_kolmogorov}]
In the following, for the notational simplicity, we often omit $|y|$ of 
$s(\abs{y})$ and $\Delta(\abs{y})$,
and just write $s$ and $\Delta$, respectively.
For contradiction, we assume that \cref{ass:AH_hardness_Estimate} does not follow, and construct a QPT algorithm $\cA$
that breaks \cref{ass:AH_hardness_GapK}.
For an arbitrary constant $k>0$, there exists a constant 
$q>0$ such that 
\begin{align}
    \frac{1}{n^q}+2^{-\Delta/3}\leq \frac{1}{n^k}
\end{align}
for all sufficiently large $n\in\N$.
Because we assume that \cref{ass:AH_hardness_Estimate} does not hold,
for any $q>0$ and for any QPT algorithm $\cQ$ there exists a QPT algorithm $\mathsf{Estimate}$ such that
\begin{align}\label{AH_hardness_Estimate_broken}
    \Pr_{y\la\cQ(1^n)}\left[\frac{99}{100}\Pr[y\la\cQ(1^n)]\leq \mathsf{Estimate}(y)\leq \frac{100}{99} \Pr[y\la\cQ(1^n)]\right]> 1-\frac{1}{n^{q}}
\end{align}
for infinitely many $n\in\N$.

Our QPT algorithm $\cA$ that solves $\mathsf{GapK}[s-\Delta,s]$ is constructed as follows:
It receives an instance $y\gets\cQ(1^n)$ and runs $\mathsf{Estimate}(y)$.
$\cA$ outputs $\mathsf{yes}$ indicating $y\in\cL_{\mathsf{yes}}$ 
if $\mathsf{Estimate}(y)\geq 2^{-s+\Delta/2}$,
and outputs $\mathsf{no}$ otherwise indicating $y\in\cL_{\mathsf{no}}$.

We use the following \cref{claim:high_prob_low_kolmogorov,claim:low_prob_high_kolmogorov}, which we will prove later.
\begin{claim}\label{claim:low_prob_high_kolmogorov}
    For all sufficiently large $n\in\N$,
        \begin{align}
            \Pr_{y \leftarrow \cQ(1^n)}\left[\Pr[y\la\cQ(1^n)] < \frac{100}{99}\cdot 2^{-s+\Delta/2 }\wedge K(y) \le s - \Delta \right] \le 2^{-\Delta/3}.
        \end{align}
\end{claim}

\begin{claim}\label{claim:high_prob_low_kolmogorov}
We have
\begin{align}\label{eq:high_prob_low_kolmogorov}
    \Pr_{y\la\cQ(1^n)}\left[ \Pr[y\la\cQ(1^n)]\geq \frac{99}{100}\cdot 2^{-s+\Delta/2}\wedge K(y)\geq s \right]=0
\end{align}    
for all sufficiently large $n\in\N$.
\end{claim}

Now, we have
\begin{align}
&\Pr_{y\la\cQ(1^n)}[\mathsf{no}\la\cA(y) \wedge y\in\cL_\mathsf{Yes}] 
+\Pr_{y\la\cQ(1^n)}[\mathsf{yes}\la\cA(y) \wedge y\in\cL_{\mathsf{No}}]\\
&=\Pr_{y\la\cQ(1^n)}[\mathsf{Estimate}(y)< 2^{-s+\Delta/2} \wedge K(y)\leq s-\Delta ] 
+\Pr_{y\la\cQ(1^n)}[\mathsf{Estimate}(y)\geq  2^{-s+\Delta/2} \wedge K(y)\geq s ]\\
&< \frac{1}{n^q}+\Pr_{y\la\cQ(1^n)}\left[\Pr[y\la\cQ(1^n)]<\frac{100}{99}\cdot 2^{-s+\Delta/2}\wedge K(y)\leq s-\Delta \right]\\
&+\Pr_{y\la\cQ(1^n)}\left[\Pr[y\la\cQ(1^n)]\geq \frac{99}{100}\cdot 2^{-s+\Delta/2}\wedge K(y)\geq s \right]\\
&\leq \frac{1}{n^q} + 2^{-\Delta/3}\\
&\leq \frac{1}{n^k},
\end{align}
for infinitely many $n\in\N$, where, in the first inequality, we have used \cref{AH_hardness_Estimate_broken}, that is,
$\frac{99}{100}\Pr[y\la\cQ(1^n)]\leq \mathsf{Estimate}(y)\leq \frac{100}{99}\Pr[y\la\cQ(1^n)]$ with probability at least $1-\frac{1}{n^q}$
and, in the second inequality, we have used \cref{claim:high_prob_low_kolmogorov,claim:low_prob_high_kolmogorov}.

\begin{proof}[Proof of \cref{claim:low_prob_high_kolmogorov}]
    Let 
   \begin{align}
   \mathsf{Low}\seteq \left\{ y\in\bit^n: K(y)\leq s-\Delta\mbox{\,\,and\,\,} \Pr[y\la\cQ(1^n)] <\frac{100}{99}2^{-s+\Delta/2}\right\}.
   \end{align} 
    Because the number of string $y\in\bit^n$ such that $K(y)\leq s-\Delta$ is at most $2^{s-\Delta+1}$, we have
    \begin{align}
        \abs{\mathsf{Low}}\leq 2^{s -\Delta+1}.
    \end{align}
    Therefore, we have
    \begin{align}
        \Pr_{y\la\cQ(1^n)}[y\in \mathsf{Low} ]
        &=\sum_{y\in \mathsf{Low}} \Pr[y\la\cQ(1^n)]\\
        &\leq \sum_{y\in \mathsf{Low}} \frac{100}{99} 2^{-s+\Delta/2}\\
        &\leq  2^{s -\Delta+1} \cdot\frac{100}{99} 2^{-s+\Delta/2}\\
        &\leq 2^{-\Delta/3}
    \end{align}
    for all sufficiently large $n\in\N$,    
    which shows the claim.
\end{proof}

\begin{proof}[Proof of \cref{claim:high_prob_low_kolmogorov}]
    Let 
    \begin{align}
    \mathsf{High}\seteq \left\{y\in\bit^n: \Pr[y\la\cQ(1^n)]\geq \frac{99}{100}2^{-s+\Delta/2}\right\}. 
    \end{align}
    Then, we have $\abs{\mathsf{High}}\leq \frac{100}{99}2^{s-\Delta/2}$.
    There exists a Turing machine $\cM$ that generates any $y\in \mathsf{High}$ by specifying the code of $\cQ$, $n$, $\frac{99}{100}2^{-s+\Delta/2}$ and the index $i$ of $y\in \mathsf{High}$\footnote{ 
    An inefficient Turing machine $y\la\cM\left(\cQ,n,\frac{99}{100}2^{-s+\Delta/2},i\right)$ works as follows:
    For all $y\in\bit^n$, $\cM$ computes $\Pr[y\la\cQ(1^n)]$ and adds $y\in\mathsf{High}$ if $\Pr[y\la\cQ(1^n)]\geq \frac{99}{100}2^{-s+\Delta/2}$.
    $\cM$ outputs the $i$-th string that belongs to $\mathsf{High}$.
    }.
    The code of $\cQ$ and $\cM$ are described by constant bits, $n$ and $\frac{100}{99}2^{-s+\Delta/2} $ are described by $O(\log(n))$ bits, and the index $i$ of $y\in \mathsf{High}$ is described by $\left(s-\Delta/2+1\right)$-bits.
    Therefore, the Kolmogorov complexity of $y\in \mathsf{High}$ is at most 
    \begin{align}
        O(1)+O(\log(n)) +s-\Delta/2+1\leq s-w(\log(n))
    \end{align}
    for all sufficiently large $n\in\N$.
    Hence
\cref{eq:high_prob_low_kolmogorov} is obtained for all sufficiently large $n\in\N$.
\end{proof}

\end{proof}

\subsection{Proof of \cref{thm:estimate}}\label{sec:estimate}

\begin{proof}[Proof of \cref{thm:estimate}]
We show \cref{thm:estimate} by showing its contraposition.
More precisely, assume that OWPuzzs do not exist.
Then, for an arbitrary QPT algorithm $\cQ$ and arbitrary two constants $c>1$ and $q>0$, we construct a QPT algorithm $\mathsf{Estimate}$ such that for infinitely many $n\in\N$, 
\begin{align}
    \Pr_{y\la\cQ(1^n)}\left[ \frac{1}{c}\Pr[y\la\cQ(1^n)]\leq \mathsf{Estimate}(y)\leq c \Pr[y\la\cQ(1^n)] \right]\geq 1-\frac{1}{n^q}.
\end{align}

In the following, we often omit $1^n$ of $\cQ(1^n)$.
From the assumption that OWPuzzs do not exist and \cref{thm:dist_invert}, there exists a QPT algorithm $\mathsf{Ext}$ such that 
\begin{align}
    \mathsf{SD}\left( (y_1,...,y_i)_{(y_1,...,y_i)\gets\cQ} , 
    (y_1...,y_{i-1}, \mathsf{Ext}(i,y_1,..., y_{i-1}) )_{(y_1,...,y_{i-1})\gets \cQ} \right)\leq \frac{1}{n^{50q}} 
\end{align}
for all $i\in[n]$
and infinitely many $n\in\N$.

We construct $\mathsf{Estimate}$ by using the $\mathsf{Ext}$ as follows.
\begin{description}
    \item[Construction of $\mathsf{Estimate}$:] $ $
    \begin{enumerate}
    \item Receive $(y_1,...,y_n)$ as input.
        \item For each $i\in[n]$, run as follows:
        \begin{itemize}
            \item Run $b\gets \mathsf{Ext}(i,y_1\cdots y_{i-1})$ for $ n^{100} n^{100q}$ times.
            Let $\mathsf{Count}_{y_1\cdots y_{i-1}}(b)$ be the number of times that $\mathsf{Ext}(i,y_1\cdots y_{i-1})$ outputs $b$.
            \item 
            Set 
            \begin{align}
            \widetilde{p}[y_i]\seteq \frac{\mathsf{Count}_{y_1\cdots y_{i-1}}(y_i)}{n^{100}n^{100q}}.
            \end{align}
        \end{itemize}
        \item Output the value of $\prod_{i=1}^n\widetilde{p}[y_i]$.
    \end{enumerate}
\end{description}

In the following, we prove that 
\begin{align}
\frac{1}{c}\Pr[y\la\cQ]\leq \prod_{i=1}^n\widetilde{p}[y_i]\leq c\cdot\Pr[y\la\cQ]    
\end{align}
with high probability over $y\la\cD_n$.

For showing \cref{thm:estimate}, we use the following \cref{claim:samp_is_high,claim:invert_is_high,claim:Hoeff}, which we prove later.
Here, $\Pr[(y_1,...,y_{i-1})\gets\cQ]=1$ if $i=1$.

\begin{claim}\label{claim:samp_is_high}
    For any real $a>0$, we have
    \begin{align}
     \Pr_{y\la\cQ}\left[\frac{1}{2a} \leq  \frac{\Pr[(y_1,..., y_{i-1}y_i)\la\cQ]}{\Pr[(y_1,..., y_{i-1})\la\cQ]} \mbox{\,\,for\,\,all } i\in [n]\right]\geq 1-\frac{n}{a}
    \end{align}
    for all $n\in\N$.
\end{claim}

\begin{claim}\label{claim:invert_is_high}
   For any real $b>0$, we have
   \begin{align}
       \Pr_{y\la\cQ}\left[\abs{\Pr[y_i\la \mathsf{Ext}(i,y_1,..., y_{i-1})]-\frac{\Pr[(y_1,..., y_{i-1} y_i)\la \cQ]}{\Pr[(y_1,..., y_{i-1})\la\cQ]}}\leq bn^{-50q}\mbox{\,\,for all\,\,}i\in[n] \right]\geq 1- \frac{n}{b}
   \end{align}
   for infinitely many $n\in\N$.
\end{claim}

\begin{claim}\label{claim:Hoeff}
    For any real $d>0$, we have
    \begin{align}
        \Pr[\abs{\widetilde{p}[y_i]-\Pr[y_i\la\mathsf{Ext} (i,y_1,...,y_{i-1})]}\leq \frac{1}{d}\mbox{ for all }i\in[n]]\geq 1-2n\exp{\frac{-2n^{100+100q}}{d^2}}
    \end{align}
    for all $n\in\N$,
    where the probability is taken over $\mathsf{Estimate}(y)$ for computing $\widetilde{p}[y_i]$.
\end{claim}

We use the claims above by setting $a=n^{q+2}$, $b=n^{q+4}$ and $d=n^{q+4}$.
From \cref{claim:invert_is_high,claim:Hoeff}, with probability at least $1-2n^{-q-3}$, we have
\begin{align}
    \abs{\widetilde{p}[y_i]-\frac{\Pr[(y_1,..., y_{i-1}y_i)\la\cQ]}{\Pr[(y_1,..., y_{i-1})\la\cQ]}}\leq 1/d + bn^{-50q}\leq 2n^{-q-4}
\end{align}
for all $i\in[n]$.
This implies that
\begin{align}
    \widetilde{p}[y_i]\leq \frac{\Pr[(y_1,...,y_i)\la\cQ]}{\Pr[(y_1,..., y_{i-1})\la\cQ]}\left(1+ 2n^{-q-4} \frac{\Pr[(y_1,..., y_{i-1})\la\cQ]}{\Pr[(y_1,..., y_i)\la\cQ]} \right)
\end{align}
with probability at least $1-2n^{-q-3}$.
Furthermore, from \cref{claim:samp_is_high}, with probability at least $1-n^{-q-1} $, we have
\begin{align}
  \frac{\Pr[(y_1,..., y_{i-1})\la\cQ]}{\Pr[(y_1,..., y_{i-1}y_i)\la\cQ]}\leq 2a=2n^{q+2}
\end{align}
for all $i\in[n]$.
Therefore, with probability at least $1-3n^{-q-1}$, we have
\begin{align}
    \widetilde{p}[y_i]\leq \frac{\Pr[(y_1,..., y_i)\la\cQ]}{\Pr[(y_1,..., y_{i-1})\la\cQ]}\left(1+ \frac{4}{n^2} \right).
\end{align}
Therefore, with probability at least $1-3n^{-q-1}$, we have
\begin{align}
    \prod_{i\in[n]}\widetilde{p}[y_i]&\leq \prod_{i\in[n]}\left(\frac{\Pr[(y_1,..., y_i)\la\cQ]}{\Pr[(y_1,..., y_{i-1})\la\cQ]}\left(1+ \frac{4}{n^2} \right)\right)\\
    &=\Pr[(y_1,..., y_n)\la\cQ] \left(1+ \frac{4}{n^2} \right)^n
\end{align}
for infinitely many $n\in\N$.\footnote{Remember that \cref{claim:invert_is_high} satisfied only for infinitely many $n\in\N$.}
For any constant $c>1$, there exists an $n_0\in\N$ such that 
\begin{align}
    \left(1+ \frac{4}{n^2} \right)^n\leq c
\end{align}
for all $n\ge n_0$. 
Therefore, we have
\begin{align}\label{one}
    \prod_{i\in[n]}\widetilde{p}[y_i]\leq c\cdot\Pr[(y_1,..., y_n)\la\cQ]
\end{align}
for infinitely many $n\in\N$.

In the same way, we can prove that, with probability at least $1-3n^{-q-1}$,
\begin{align}
    \frac{1}{c}\Pr[(y_1,..., y_n)\la\cQ]\leq \prod_{i\in[n]}\widetilde{p}[y_i]
\end{align}
for infinitely many $n\in\N$ as follows.
From \cref{claim:invert_is_high,claim:Hoeff}, with probability at least $1-2n^{-q-3}$, we have
\begin{align}
     \frac{\Pr[(y_1,..., y_i)\la\cQ]}{\Pr[(y_1,..., y_{i-1})\la\cQ]}\left(1 - 2n^{-q-4} \frac{\Pr[(y_1,..., y_{i-1})\la\cQ]}{\Pr[(y_1,..., y_i)\la\cQ]} \right)\leq \widetilde{p}[y_i]
\end{align}
for infinitely many $n\in\N$ and for all $i\in[n]$.
Furthermore, from \cref{claim:samp_is_high}, with probability at least $1-n^{-q-1} $, we have
\begin{align}
 -2n^{q+2}=-2a\leq  -\frac{\Pr[(y_1,..., y_{i-1})\la\cQ]}{\Pr[(y_1\cdots y_{i-1}y_i)\la\cQ]}.
\end{align}
Therefore, with probability at least $1-3n^{-q-1}$, we have
\begin{align}
    \frac{\Pr[(y_1,..., y_i)\la\cQ]}{\Pr[(y_1,..., y_{i-1})\la\cQ]}\left(1 - \frac{4}{n^2} \right) \leq \widetilde{p}[y_i].
\end{align}
Hence, with probability at least $1-3n^{-q-1}$, we have
\begin{align}
    \prod_{i\in[n]}\widetilde{p}[y_i]&\geq \prod_{i\in[n]}\left(\frac{\Pr[(y_1,..., y_i)\la\cQ]}{\Pr[(y_1,..., y_{i-1})\la\cQ]}\left(1- \frac{4}{n^2} \right)\right)\\
    &=\Pr[(y_1,..., y_n)\la\cQ] \left(1- \frac{4}{n^2} \right)^n
\end{align}
for infinitely many $n\in\N$.
For any constant $c>1$, there exists an $n_1\in\N$ such that 
\begin{align}
    \frac{1}{c}\leq\left(1-\frac{4}{n^2}\right)^n
\end{align}
for all $n\ge n_1$.
Therefore, we have, with probability at least $1-3n^{-q-1}$
\begin{align}\label{two}
    \frac{1}{c}\Pr[(y_1,..., y_n)\la\cQ]\leq \prod_{i\in[n]}\widetilde{p}[y_i]
\end{align}
for infinitely many $n\in\N$.
By combining \cref{one} and \cref{two}, we have that
\begin{align}
    \frac{1}{c}\Pr[(y_1,..., y_n)\la\cQ]\leq \prod_{i\in[n]}\widetilde{p}[y_i]
    \le c \Pr[(y_1,..., y_n)\la\cQ]
\end{align}
with probability at least $1-6n^{-q-1}$ (which is larger than $1-\frac{1}{n^q}$ for sufficiently large $n\in\N$)
for infinitely many $n\in\N$, which shows the lemma.

\begin{proof}[Proof of \cref{claim:samp_is_high}]
This is shown by a standard probabilistic argument.
    Let 
    \begin{align}
        \mathsf{Good}\seteq\left\{y\in\bit^n:\frac{\Pr[(y_1,..., y_{i})\la\cQ]}{\Pr[(y_1,..., y_{i-1})\la\cQ]}\geq \frac{1}{2a} \mbox{ for all }i\in[n] \right\},
    \end{align}  
    and let
    \begin{align}
        \mathsf{Bad_i}\seteq \left\{y\in\bit^n:\frac{\Pr[(y_1,..., y_i)\la\cQ]}{\Pr[(y_1,..., y_{i-1})\la\cQ]}< \frac{1}{2a} \right\}.
    \end{align}
    Because
    \begin{align}
        \sum_{y\in\mathsf{Good}}\Pr[(y_1,..., y_n)\la\cQ]\geq 1-\sum_{i\in[n]}\sum_{y\in\mathsf{Bad}_i}\Pr[(y_1,..., y_n)\la\cQ],
    \end{align}
    it is sufficient to show 
    \begin{align}
        \sum_{y\in\mathsf{Bad}_i}\Pr[(y_1,..., y_n)\la\cQ] <1/a
    \end{align}
    for all $i\in[n]$.
    We have
    \begin{align}
        \sum_{y\in\mathsf{Bad}_i}\Pr[(y_1,..., y_n)\la\cQ]
        &=\sum_{y\in\mathsf{Bad}_i} \left(\prod_{j\in[n]}\frac{\Pr[(y_1,..., y_j)\la\cQ]}{\Pr[(y_1,..., y_{j-1})\la\cQ]}\right) \\
        &=\sum_{y\in\mathsf{Bad}_i} \left(\prod_{j\in[n]\backslash i}\frac{\Pr[(y_1,..., y_j)\la\cQ]}{\Pr[(y_1,..., y_{j-1})\la\cQ]}\right)\cdot \frac{\Pr[(y_1,..., y_{i-1}y_i)\la \cQ]}{\Pr[(y_1,..., y_{i-1})\la\cQ]} \\
        &<\sum_{y\in\mathsf{Bad}_i} \left(\prod_{j\in[n]\backslash i}\frac{\Pr[(y_1\cdots y_j)\la\cQ]}{\Pr[(y_1,..., y_{j-1})\la\cQ]}\right)\cdot \frac{1}{2a} \\
        &<\sum_{y\in\bit^n} \left(\prod_{j\in[n]\backslash i}\frac{\Pr[(y_1,..., y_j)\la\cQ]}{\Pr[(y_1,..., y_{j-1})\la\cQ]}\right)\cdot \frac{1}{2a}\\
        &= \frac{1}{a}.
    \end{align}
In the last equation, we have used that
\begin{align}
    &\sum_{y_1,...,y_n\in\bit^n} \left(\prod_{j\in[n]\backslash i}\frac{\Pr[(y_1,..., y_j)\la\cQ]}{\Pr[(y_1,..., y_{j-1})\la\cQ]}\right)\\
    &=\sum_{y_1,...,y_n\in\bit^n}\frac{\Pr[(y_1,...,y_{i-1})\gets\cQ]}{\Pr[(y_1,...,y_i)\gets\cQ]}\Pr[(y_1,...,y_n)\gets\cQ]\\
    &=\sum_{y_1,...,y_n\in\bit^n}\Pr[(y_1,...,y_{i-1})\gets\cQ]
    \Pr[(y_{i+1},...,y_n)\gets\cQ|(y_1,...,y_i)\gets\cQ]\\
    &=\sum_{y_1,...,y_i\in\bit^i}\Pr[(y_1,...,y_{i-1})\gets\cQ]\\
    &=\sum_{y_i\in\bit}1\\
    &=2.
\end{align}

\end{proof}

\begin{proof}[Proof of \cref{claim:invert_is_high}]
    From the definition of $\mathsf{Ext}$, we have
    \begin{align}
    \mathsf{SD}\left( (y_1,..., y_i)_{(y_1,...,y_i)\gets\cQ} , 
    (y_1,...,y_{i-1}, \mathsf{Ext}(i,y_1,..., y_{i-1}) )_{(y_1,...,y_{i-1})\gets \cQ} \right)\leq \frac{1}{n^{50q}} 
    \end{align}
    for all $i\in[n]$.
This implies that
    \begin{align}
        &\sum_{y_1,..., y_{i-1}\in\bit^{i-1}}
        \left|\Pr[(y_1,..., y_{i-1},1)\la\cQ]-\Pr[(y_1,..., y_{i-1})\la\cQ]\Pr[1\la\mathsf{Ext} (i,y_1,..., y_{i-1})] \right|\\
        &=\sum_{y_1,..., y_{i-1}\in\bit^{i-1}}\Pr[(y_1,..., y_{i-1})\la\cQ]\cdot\abs{\frac{\Pr[(y_1,..., y_{i-1},1)\la\cQ]}{\Pr[(y_1,..., y_{i-1})\la\cQ]} - \Pr[1\la\mathsf{Ext}(i,y_1,..., y_{i-1}) ]}\\
        &= \mathbb{E}_{(y_1,..., y_{i-1})\la\cQ}\left[\abs{\frac{\Pr[(y_1,..., y_{i-1},1)\la\cQ]}{\Pr[(y_1,..., y_{i-1})\la\cQ]} - \Pr[1\la\mathsf{Ext}(i,y_1,..., y_{i-1}) ]}\right]
        \leq n^{-50q}
    \end{align}
    for all $i\in[n]$.
    From Markov inequality, for each $i\in[n]$, we have
    \begin{align}
        \Pr_{(y_1,..., y_{i-1})\la\cQ}\left[\abs{\frac{\Pr[(y_1,..., y_{i-1},1)\la\cQ]}{\Pr[(y_1,..., y_{i-1})\la\cQ]} - \Pr[1\la\mathsf{Ext}(i,y_1,..., y_{i-1}) ]}\geq bn^{-50q} \right]\leq \frac{1}{b}.
    \end{align}
    Therefore, 
    \begin{align}
        &\abs{\frac{\Pr[(y_1,.., y_{i-1},1)\la\cQ]}{\Pr[(y_1,..., y_{i-1})\la\cQ]} - \Pr[1\la\mathsf{Ext}(i,y_1,..., y_{i-1}) ]}\\
        &=\abs{\frac{\Pr[y_1,..., y_{i-1},0\la\cQ]}{\Pr[y_1,..., y_{i-1}\la\cQ]}-\Pr[0\la\mathsf{Ext}(i,y_1,..., y_{i-1}) ] }< bn^{-50q}
    \end{align}
    is satisfied for all $i\in[n]$
    with probability at least $1-\frac{n}{b}$.
\end{proof}

\begin{proof}[Proof of \cref{claim:Hoeff}]
From the Hoeffding inequality, for each $i\in[n]$,
\begin{align}
\Pr_{\widetilde{p}[y_i]\la\mathsf{Estimate}(y)}\left[ \left|\widetilde{p}[y_i] - \Pr[y_i\la \mathsf{Ext}(i,y_1,..., y_{i-1}) ]\right| > \frac{1}{d} \right]\leq  2\exp\left( \frac{-2n^{100+100q}}{d^2}\right).
\end{align}
Union bound implies that
\begin{align}
    \Pr_{\{\widetilde{p}[y_i]\}_{i\in[n]}\la\mathsf{Estimate}(y)}\left[\abs{\widetilde{p}[y_i]-\Pr[y_i\la\mathsf{Ext}(i,y_1,..., y_{i-1})]}>1/d \mbox{\,\,for some\,\,}i\in[n] \right]\leq 2n\exp\left(\frac{-2n^{100+100q}}{d^2}\right).
\end{align}
\end{proof}

\end{proof}

\subsection{Proof of \cref{thm:QEFID_GapK}}

\begin{proof}[Proof of \cref{thm:QEFID_GapK}]
Assume that a QPRG exists.
Our goal is to show that, for any real $0<\epsilon<1$, 
there exist an integer $k>0$ and a QPT algorithm $\cQ$ such that for any QPT algorithm $\cA$, 
\begin{align}
    \Pr_{x\la\cQ(1^n)}[ \mathsf{no}\la \cA(x) \wedge x\in\cL_{\mathsf{Yes}}]+\Pr_{x\la\cQ(1^n)}[\mathsf{yes}\la\cA(x)\wedge x\in\cL_{\mathsf{No}} ]\geq \frac{1}{2}- \frac{1}{n^k}
\end{align}
for all sufficiently large $n\in\N$.
Here $\cL_{\mathsf{Yes}}$ (resp. $\cL_{\mathsf{No}}$) is the set of yes (resp. no) instances of $\mathsf{GapK}[n-n^{\epsilon},n-\Delta]$.
Let $\tau\seteq\frac{1+\epsilon}{2}$ so that $\epsilon<\tau < 1$.
From \cref{lem:far}, for any real $0<\epsilon<\tau<1$, there exists a QPRG $\Gen$ such that
\begin{align}
    \mathsf{SD}((x)_{x\la\Gen(1^n)},(x)_{x\la \bit^n})\geq 1-2^{-n^{\tau}}\label{eqn:SD}
\end{align}
for all sufficiently large $n\in\N$.
Let $\cQ$ be a QPT algorithm such that, on input $1^n$, it samples $y\la \bit^n$ with probability $1/2$ and 
$x\la\mathsf{Gen}(1^n)$ with probability $1/2$.
We claim that thus defined $\cQ$ is the desired sampling algorithm.
For contradiction, assume that, for any integer $k>0$, there exists a QPT algorithm $\cA$ such that 
\begin{align}
\label{assumption_uniformQEFID_1}
    \Pr_{x\la\cQ(1^n)}[ (\mathsf{yes}\la \cA(x) \wedge x\in\cL_{\mathsf{Yes}})\vee(\mathsf{no}\la\cA(x)\wedge x\in\cL_{\mathsf{No}} )\vee (x\notin\cL_{\mathsf{yes}}\cup\cL_{\mathsf{no}})]> \frac{1}{2}+\frac{1}{n^k}
\end{align}
for infinitely many $n\in\N$. 
Then, we construct a QPT adversary $\cB$ that breaks the QPRG $\Gen$ as follows:
$\cB$ receives $x$, and runs $\mathsf{yes}/\mathsf{no}\la\cA(x) $.
$\cB$ outputs 1 if $\mathsf{yes}\la\cA(x)$ and outputs 0 otherwise.
We use the following \cref{claim:high,claim:low}.

\begin{claim}\label{claim:high}
For any $\Delta(n)=w(\log(n))$,
we have
    \begin{align}
        \Pr_{x\la \bit^n}[K(x) >n-\Delta]\geq 1-2^{-\Delta+1}.
    \end{align}
\end{claim}

\begin{claim}\label{claim:low}
For an arbitrary real $G>0$, we have
\begin{align}
    \Pr_{x\la\Gen(1^n)}[K(x)< n-n^{\tau} + \log(1/G) +O(\log(n))]\geq 1-\left(G+ 2^{-n^{\tau}}\right)
\end{align}
for infinitely many $n\in\N$.
\end{claim}

\cref{claim:high,claim:low} will be shown later.
We use \cref{claim:low} by setting $G=2^{-n^{\epsilon}}$.
Hence, we have 
\begin{align}
  1-2\cdot 2^{-n^{\epsilon}}&\leq \Pr_{x\la\Gen(1^n)}[K(x)<  n-n^{\tau}+n^{\epsilon}+O(\log(n))]\\
  &\leq \Pr_{x\la\Gen(1^n)}[K(x)<  n-3n^{\epsilon}+n^{\epsilon}+O(\log(n))]\\
  &\leq \Pr_{x\la\Gen(1^n)}[K(x)<  n-n^{\epsilon}]
\end{align}
for infinitely many $n\in\N$.
From \cref{claim:high,claim:low,assumption_uniformQEFID_1}, 
we have
\begin{align}
    \frac{1}{2}+\frac{1}{n^k}&<\Pr_{x\la\cQ(1^n)}[ (\mathsf{yes}\la \cA(x) \wedge x\in\cL_{\mathsf{Yes}}) \vee (\mathsf{no}\la\cA(x)\wedge x\in\cL_{\mathsf{No}}) \vee (x\notin\cL_{\mathsf{yes}}\cup\cL_{\mathsf{no}})]\\
    &<\frac{1}{2}\left(\Pr_{x\la\Gen(1^n)}[ \mathsf{yes}\la \cA(x)] +\Pr_{x\la \bit^n}[\mathsf{no}\la\cA(x)]+2\cdot 2^{-n^{\epsilon}}+2^{-\Delta+1}\right)\\
    &\leq \frac{1}{2}\left(\Pr_{x\la\Gen(1^n)}[ \mathsf{yes}\la \cA(x) ] +1-\Pr_{x\la \bit^n}[\mathsf{yes}\la\cA(x) ]+2\cdot 2^{-n^{\epsilon}}+2^{-\Delta+1}\right)
\end{align}
for infinitely many $n\in\N$.
This implies that
\begin{align}
   \frac{2}{n^{k}}-2\cdot 2^{-n^{\epsilon}}-2^{-\Delta+1} <
   \Pr_{x\la\Gen(1^n)}[\mathsf{yes}\la\cA(x)]-\Pr_{x\la \bit^n}[\mathsf{yes}\la\cA(x)].
\end{align}
This implies that, for some polynomial $q$,
\begin{align}
\abs{\Pr_{x\la\Gen(1^n)}[1\la\cB(1^n,x)] -\Pr_{x\la \bit^n}[1\la\cB(1^n,x)]}\geq \frac{1}{q(n)}
\end{align}
for infinitely many $n\in\N$.
This is the contradiction.

\begin{proof}[Proof of \cref{claim:high}]
    The number of string $y\in\bit^n$ such that $K(y)\leq n-\Delta$ is at most $2^{n-\Delta+1}$.
    Therefore, we have
    \begin{align}
        \Pr_{x\la \bit^n}[K(x)\leq n-\Delta]\leq \frac{2^{n-\Delta+1}}{2^n}=2^{-{\Delta}+1}.
    \end{align}
\end{proof}

\begin{proof}[Proof of \cref{claim:low}]
    We define the following sets.
    \begin{align}
        &A\seteq\{x\in\bit^n:\Pr[x\la\Gen(1^n)]<2^{-n}\}\\
        &B\seteq \{x\in\bit^n: 2^{-n}\leq \Pr[x\la\Gen(1^n)]< G\cdot 2^{-n+n^{\tau}}\}\\
        &C\seteq\{x\in\bit^n:G\cdot 2^{-n+n^{\tau}}\leq \Pr[x\la\Gen(1^n)]\leq 1\}.
    \end{align}
    From \cref{eqn:SD}, we have
    \begin{align}
        1- 2^{-n^{\tau}}&\leq \sum_{x\in A}(\Pr[x\la \bit^n ] - \Pr[x\la\Gen(1^n)])\\
        & =\sum_{x\in A} (2^{-n} - \Pr[x\la\Gen(1^n)])\\
        &\leq 2^{-n} \abs{A}.
    \end{align}
    This implies that
    \begin{align}
        \abs{B}+\abs{C}\leq 2^{n-n^{\tau}}\label{eqn:count}.
    \end{align}
    Furthermore, we have
    \begin{align}
        1- 2^{-n^{\tau}}&\leq \sum_{x\in B}(\Pr[x\la\Gen(1^n)]-2^{-n})+\sum_{x\in C}(\Pr[x\la\Gen(1^n)]-2^{-n})\\
        &\leq \sum_{x\in B}G\cdot 2^{-n+n^{\tau}}+\sum_{x\in C}\Pr[x\la\Gen(1^n)] \\
        &\leq \abs{B}G\cdot 2^{-n+n^{\tau}} +\sum_{x\in C}\Pr[x\la\Gen(1^n)].
    \end{align}
    This and \cref{eqn:count} implies that
    \begin{align}
       1- G- 2^{-n^{\tau}} &\leq \sum_{x\in C}\Pr[x\la\Gen(1^n)].
    \end{align}
    Now, the remaining part is to show that 
    \begin{align}
        K(x)< n-n^{\tau} +\log(1/G) +O(\log(n))
    \end{align}
    for all $x\in C$.
    There exists a Turing machine $\cM$ such that it generates $x\in C$ by taking as input the $\Gen$, $n$, $\frac{1}{G}\cdot 2^{n-n^{\tau}}$
    and the index $i$ 
    of $x\in C$.~\footnote{ 
    A Turing machine $y\la\cM\left(\Gen,n,\frac{1}{G}2^{n-n^\tau},i\right)$ works as follows:
    For all $y\in\bit^n$, $\cM$ computes $\Pr[y\la\Gen(1^n)]$ and adds $y\in\mathsf{High}$ if $\Pr[y\la\Gen(1^n)]\geq G\cdot 2^{-n+n^{\tau}}$.
    $\cM$ outputs the $i$-th string that belongs to $\mathsf{High}$.
    }
    Therefore, $K(x)$ is upper bounded by
    \begin{align}
        O(1)+O(\log(n)) +\log(\frac{1}{G}\cdot 2^{n-n^{\tau}} )\leq n-n^{\tau} +\log(1/G)+O(\log(n))
    \end{align}
    for all sufficiently large $n\in\N$.
\end{proof}
\end{proof}

\subsection{Proof of \cref{thm:kolmogorov_from_owpuzz}}
We can show \cref{thm:kolmogorov_from_owpuzz} in the almost same way as \cref{thm:QEFID_GapK} by considering nuQPRGs instead of QPRGs.
Therefore, instead of re-providing the proof, we briefly explain the difference.

In the proof of \cref{thm:QEFID_GapK}, we consider QPRGs $\Gen(1^n)$ given in \cref{lem:far} and construct 
a QPT algorithm $\cQ$ such that, on input $1^n$, it outputs $y\la\Gen(1^n)$ with probability $1/2$ and $y\la\bit^n$ with probability $1/2$.
We showed the strongly quantum hardness of $\mathsf{GapK}$ over $x\gets\cQ$ 
by using the fact that $K(x)$ is sufficiently small with overwhelming probability over the distribution $x\la\Gen(1^n)$ when $\Gen(1^n)$ is exponentially far from the uniform distribution.

For the proof of \cref{thm:kolmogorov_from_owpuzz}, we consider nuQPRGs $\Gen(1^n,\mu)$ given in \cref{lem:non_uniform_far} instead of considering QPRGs.
Then, we construct a QPT algorithm $\cQ$ that, on input $1^n$, samples $\mu^*\la[n]$, and outputs $y\la \Gen(1^n,\mu^*)$ with probability $1/2$ and $y\la\bit^n$ with probability $1/2$.
Now, we have $K(x)$ is sufficiently small with overwhelming probability over the distribution $x\la\Gen(1^n,\mu^*)$.
Furthermore, such $\mu^*$ is sampled with probability $1/n$. 
Therefore, we can show the weakly-quantum-average-hardness of $\mathsf{GapK}$ over $x\gets\cQ$.

\section{Results on IV-PoQ}

\begin{theorem}
        If there exists a real $0<\epsilon<1$ and a polynomial-time-computable function $\Delta(n)=w(\log(n))$ such that \cref{ass:c_AH_hardness_GapK} holds with $s_1=n-n^{\epsilon}$ and $s_2=n-\Delta$,
        then
        the QAS/OWF condition holds.
\end{theorem}

This theorem is obtained from the following lemmas.\footnote{A proof of \cref{thm:c_owpuzz_from_kolmogorov} is almost the same as that of \cref{thm:owpuzz_from_kolmogorov}, which is a slighly stonger version of \cite{IRS21}, and therefore
we omit its proof.}

\begin{lemma}[\cite{IRS21}]\label{thm:c_owpuzz_from_kolmogorov}
    Suppose that there exists a real $0<\epsilon<1$ and a polynomial-time-computable function $\Delta(n)=w(\log(n))$ such that \cref{ass:c_AH_hardness_GapK} holds with $s_1=n-n^{\epsilon}$ and $s_2=n-\Delta$.
    Then, \cref{ass:c_AH_hardness_Estimate} holds.
\end{lemma}

\begin{lemma}\label{thm:c_estimate}
       \cref{ass:c_AH_hardness_Estimate} implies the QAS/OWF condition.
\end{lemma}

\subsection{Proof of \cref{thm:c_estimate}}

\begin{proof}[Proof of \cref{thm:c_estimate}]

    For contradiction, suppose that the QAS/OWF condition is not satisfied.
    Our goal is to show that for any real $c>1$, integer $q>0$, and QPT algorithm $\cQ$, 
    there exists a PPT algorithm $\mathsf{Estimate}$ such that 
    \begin{align}
        \Pr_{x\la\cQ(1^n)}\left[\frac{1}{c}\Pr[x\la\cQ(1^n)]\leq \mathsf{Estimate}(x)\leq c\Pr[x\la\cQ(1^n)]\right]\geq 1-\frac{1}{n^q}
    \end{align}
    for infinitely many $n\in\N$.

    From the assumption that the QAS/OWF condition is not satisfied and \cref{lem:shirakawa}, 
    for any real $k>0$, there exists a PPT algorithm $\cS$ such that 
    \begin{align}\label{chikai}
    \mathsf{SD}((x)_{x\la \cQ(1^n)},\cS(1^n;r)_{r\la \bit^{t(n)}})\leq \frac{1}{n^k}
    \end{align}
    for infinitely many $n\in\N$. 
    Here $t(n)$ is the length of the random seed of $\cS(1^n)$.
    Let $\Sigma$ be the set of such $n$.
    
    For each $n\in\mathbb{N}$, define a function $f_n:\bit^{m(n)}\to\bit^*$ as follows:
    \begin{enumerate}
        \item 
    Receive $(r,i)\in\bit^{t(n)}\times[n-1]$ as input.\footnote{Actually, the integer $i$ is encoded into a bit string. If the input bit string
    corresponds to an integer that is outside of $[n-1]$, $f_n$ outputs just 0.}
    \item 
    Compute $x=\cS(1^n;r)$.
    \item 
    Output $i$ and the first $i$ bits, $x_1,...,x_i$, of $x$.
    \end{enumerate}
 From \cref{lem:shirakawa}, there exists a PPT algorithm $\cR$ such that 
    \begin{align}
        \mathsf{SD}\left(\{(r,i),f_n(r,i)\}_{r\la\bit^{t(n)},i\gets[n-1]},
        \{\cR(1^{m(n)},f_n(r,i)),f_n(r,i)\}_{r\la\bit^{t(n)},i\gets[n-1]} 
        \right)\leq \frac{1}{n^k}
    \end{align}
for all $n\in\Sigma$.

Now, we define the following PPT algorithm, $\mathsf{Ext}$, as follows:
    \begin{enumerate}
        \item Receive $1^n$, $i\in[n-1]$, and $(y_1,...,y_i)\in\bit^i$ as input.
        \item Run $(r',i')\la \cR(1^{m(n)},i,y_1,...,y_i)$.
        \item Run $x=\cS(1^n;r')$, and output $x_{i+1}$.
    \end{enumerate}
We claim that $\mathsf{Ext}$ satisfies
\begin{align}\label{extrapolate}
    \mathsf{SD}( (x_1,...,x_{i+1})\la \cS(1^n), 
    (x_1,...,x_i,\mathsf{Ext}(1^n,i,x_1,...,x_i))_{x_1,...,x_i\la\cS(1^n)}  )<\frac{1}{n^{k-1}}
\end{align}
for all $i\in[n-1]$ and all $n\in\Sigma$.
This is shown as follows.
    \begin{align}
        \frac{1}{n^k}
        &\ge
        \mathsf{SD}\left(\{(r,i),f_n(r,i)\}_{r\la\bit^{t(n)},i\gets[n-1]},
        \{\cR(1^{m(n)},f_n(r,i)),f_n(r,i)\}_{r\la\bit^{t(n)},i\gets[n-1]} 
        \right)\\
         &=
        \mathsf{SD}\left(\{(r,i),(i,x_1,...,x_i)\}_{i\gets[n-1],r\gets\bit^{t(n)},(x_1,...,x_i)=\cS(1^n;r)},\right.\\
        &\left.\{\cR(1^{m(n)},(i,x_1,...,x_i)),(i,x_1,...,x_i)\}_{i\gets[n-1],r\gets\bit^{t(n)},(x_1,...,x_i)=\cS(1^n;r)} \right)\\
          &\ge
        \mathsf{SD}\left(\{(r,i),(i,x_1,...,x_i),x_{i+1}\}_{i\gets[n-1],r\gets\bit^{t(n)},(x_1,...,x_i)=\cS(1^n;r)},\right.\\
        &\left.\{\cR(1^{m(n)},(i,x_1,...,x_i)),(i,x_1,...,x_i),\mathsf{Ext}(1^n,i,x_1,...,x_i)\}_{i\gets[n-1],r\gets\bit^{t(n)},(x_1,...,x_i)=\cS(1^n;r)} \right)\\
           &\ge
        \mathsf{SD}\left(\{i,x_1,...,x_i,x_{i+1}\}_{i\gets[n-1],(x_1,...,x_i,x_{i+1})\gets\cS(1^n)},\right.\\
        &\left.\{i,x_1,...,x_i,\mathsf{Ext}(1^n,i,x_1,...,x_i)\}_{i\gets[n-1],(x_1,...,x_i)\gets\cS(1^n)} \right)\\
            &\ge
        \frac{1}{n-1}\sum_{i=1}^{n-1}
        \mathsf{SD}\left(\{x_1,...,x_i,x_{i+1}\}_{(x_1,...,x_i,x_{i+1})\gets\cS(1^n)},\right.\\
        &\left.\{x_1,...,x_i,\mathsf{Ext}(1^n,i,x_1,...,x_i)\}_{(x_1,...,x_i)\gets\cS(1^n)} \right).
    \end{align}

Then we claim that
\begin{align}
    \mathsf{SD}( (x_1,...,x_{i+1})\la \cQ(1^n), 
    (x_1,...,x_i,\mathsf{Ext}(1^n,i,x_1,...,x_i))_{x_1,...,x_i\la\cQ(1^n)}  )<\frac{1}{n^{k-2}}
\end{align}
for all $i\in[n-1]$ and all $n\in\Sigma$.
This is shown as follows
from the triangle inequality, \cref{extrapolate}, and \cref{chikai}.
\begin{align}
    &\mathsf{SD}( (x_1,...,x_{i+1})\la \cQ(1^n), 
    (x_1,...,x_i,\mathsf{Ext}(1^n,i,x_1,...,x_i))_{x_1,...,x_i\la\cQ(1^n)}  )\\
    &< \mathsf{SD}( (x_1,...,x_{i+1})\la \cQ(1^n), (x_1,...,x_{i+1})\gets\cS(1^n))\\
    &+\mathsf{SD}((x_1,...,x_{i+1})\gets\cS(1^n),(x_1,...,x_i,\mathsf{Ext}(1^n,i,x_1,...,x_i))_{x_1,...,x_i\la\cS(1^n)}  )\\
    &+\mathsf{SD}((x_1,...,x_i,\mathsf{Ext}(1^n,i,x_1,...,x_i))_{x_1,...,x_i\la\cS(1^n)},
    (x_1,...,x_i,\mathsf{Ext}(1^n,i,x_1,...,x_i))_{x_1,...,x_i\la\cQ(1^n)})\\
    &\le\frac{1}{n^k}+\frac{1}{n^{k-1}}+\frac{1}{n^k}\\
    &\le\frac{1}{n^{k-2}}.
\end{align}

By using $\mathsf{Ext}$, we construct a PPT algorithm $\mathsf{Estimate}$ as follows:
    \begin{enumerate}
    \item Receive $(y_1,...,y_n)\in\bit^n$ as input.
        \item For each $i\in[n]$, run as follows:
        \begin{itemize}
            \item Run $b\gets \mathsf{Ext}(i,y_1\cdots y_{i-1})$ for $ n^{100} n^{100q}$ times.
            Let $\mathsf{Count}_{y_1\cdots y_{i-1}}(b)$ be the number of times that $\mathsf{Ext}(i,y_1\cdots y_{i-1})$ outputs $b$.
            \item 
            Set 
            \begin{align}
            \widetilde{p}[y_i]\seteq \frac{\mathsf{Count}_{y_1\cdots y_{i-1}}(y_i)}{n^{100}n^{100q}}.
            \end{align}
        \end{itemize}
        \item Output the value of $\prod_{i=1}^n\widetilde{p}[y_i]$.
    \end{enumerate}
The analysis of $\mathsf{Estimate}$ is almost the same as \cref{thm:estimate},
and therefore we omit it.
\end{proof}

\ifnum\anonymous=0
\paragraph{Acknowledgements.}
TH is supported by 
JSPS research fellowship and by JSPS KAKENHI No. JP22J21864.
TM is supported by
JST CREST JPMJCR23I3,
JST Moonshot R\verb|&|D JPMJMS2061-5-1-1, 
JST FOREST, 
MEXT QLEAP, 
the Grant-in Aid for Transformative Research Areas (A) 21H05183,
and 
the Grant-in-Aid for Scientific Research (A) No.22H00522.
\else
\fi

\ifnum\llncs=1
\bibliographystyle{alpha} 
\bibliography{abbrev3,crypto,reference}
\else
\bibliographystyle{alpha} 
\bibliography{abbrev3,crypto,reference}
\fi

\appendix

\appendix
\section{Proof of \cref{lem:far}}\label{sec:appendix}

\begin{proof}[Proof of \cref{lem:far}]
Suppose that QEFID with uniform distribution exists.
Then, there exists $\Gen$ such that
\begin{align}
    \mathsf{SD}((x)_{x\la\Gen(1^n)},(x)_{x\la \bit^n})\geq \frac{1}{n}
\end{align}
for all sufficiently large $n\in\N$.

For an arbitrary real $0<\tau <1$, we consider the following $\Gen^*$.
\begin{description}
    \item[The description of $\Gen^*(1^{n})$:]$ $
    \begin{itemize}
        \item Let $A=n^{\frac{1-\tau}{2}}$ and $B=n^{\frac{1+\tau}{2}}$.
        \item Run $ x_i\la \Gen(1^A)$ for all $i\in[B]$.
        \item Output $\{x_i\}_{i\in[B]}$.
    \end{itemize} 
\end{description}
We have
\begin{align} 
    \mathsf{SD}\left( (\{x_i\}_{i\in[B]})_{\{x_i\}_{i\in[B]}\la \Gen^*(1^n)},X_{X\la\bit^{n}}  \right)&\geq 1- \exp\left(- B\cdot\mathsf{SD}\left(x_{x\la\Gen(1^A)},x_{x\la\bit^A}\right) \right)  \\
    &\geq 1-\exp(-\frac{B}{A})=1-\exp(-n^{\tau})>1-2^{-n^\tau}.
\end{align}
Furthermore, computationally indistinguishability of $\Gen^*$ follows from a standard argument.
\end{proof}

\ifnum\cameraready=1
\else
\ifnum\submission=1
\newpage
\setcounter{tocdepth}{1}
\tableofcontents
\else
\fi
\fi

\end{document}